\documentclass[11pt]{article}

\usepackage[margin=1in]{geometry}
\usepackage{amsmath,amsthm,amssymb}
\usepackage{environ, tabularx}
\usepackage{tikz-cd,graphicx}
\usepackage{float}
\usepackage{indentfirst}
\usepackage{mathtools}
\usepackage{hyperref}
\usepackage{amsmath}
\usepackage{dirtytalk}
\usepackage{pifont}
\usepackage{soul}
\usepackage{xcolor}
\usepackage{mathrsfs}

\newcommand{\bd}{\text{Bd}}

\newcommand{\R}{\mathbb{R}}
\newcommand{\Z}{\mathbb{Z}}
\newcommand{\N}{\mathbb{N}}

\newcommand{\K}{\mathcal{K}}
\newcommand{\M}{\mathcal{M}}
\newcommand{\C}{\mathcal{C}}
\newcommand{\D}{\mathcal{D}}
\renewcommand{\H}{\mathcal{H}}

\newcommand{\NP}{\textsf{NP}}

\newcommand{\supp}{\textsf{supp}}

\newcommand{\mobius}{\text{M\"{o}bius}}
\newcommand{\dist}{\textsf{dist}}

\newtheorem{proposition}{Proposition}[section]
\newtheorem{definition}{Definition}

\makeatletter

\usepackage{aliascnt}
\newtheorem{theorem}{Theorem}[section]

\newaliascnt{lemma}{theorem}
\newtheorem{lemma}[lemma]{Lemma}
\aliascntresetthe{lemma}

\newaliascnt{corollary}{theorem}
\newtheorem{corollary}[corollary]{Corollary}
\aliascntresetthe{corollary}

\DeclareMathOperator{\im}{im}

\DeclareMathOperator{\Bd}{Bd}

\author{William Maxwell \and Amir Nayyeri}

\title{Generalized max-flows and min-cuts in simplicial complexes}
\date{}

\begin{document}

\maketitle

\begin{abstract}
We consider high dimensional variants of the maximum flow and minimum cut problems in the setting of simplicial complexes and provide both algorithmic and hardness results.
By viewing flows and cuts topologically in terms of the simplicial (co)boundary operator we can state these problems as linear programs and show that they are dual to one another.
Unlike graphs, complexes with integral capacity constraints may have fractional max-flows.
We show that computing a maximum integral flow is \NP-hard.
Moreover, we give a combinatorial definition of a simplicial cut that seems more natural in the context of optimization problems and show that computing such a cut is \NP-hard.
However, we provide conditions on the simplicial complex for when the cut found by the linear program is a combinatorial cut.
For $d$-dimensional simplicial complexes embedded into $\R^{d+1}$ we provide algorithms operating on the dual graph: computing a maximum flow is dual to computing a shortest path and computing a minimum cut is dual to computing a minimum cost circulation.
Finally, we investigate the Ford-Fulkerson algorithm on simplicial complexes, prove its correctness, and provide a heuristic which guarantees it to halt.
\end{abstract}

\section{Introduction}
Computing flows and cuts are fundamental algorithmic problems in graphs, which are one dimensional simplicial complexes.
In this paper,  we explore generalizations of these algorithmic problems in higher dimensional simplicial complexes.
Flows and cuts in simplicial complexes have natural algebraic definitions arising from the theory of simplicial (co)homology.
A flow is an element of the kernel of the simplicial boundary operator, and a cut is an element of the image of the simplicial coboundary operator.
These subspaces serve as generalizations of the cycle and cut spaces of a graph.
This generalization has been studied by Duval, Klivans, and Martin in the setting of CW complexes~\cite{Duval2015}.
We formulate the algorithmic problems of computing max-flows and min-cuts algebraically.
By forgetting about the underlying graph structure and focusing on the (co)boundary operators, we obtain methods that naturally generalize to high dimensions.

In a graph an $st$-flow is an assignment of real values to the edges satisfying the conservation of flow constraints:
the net flow out of any vertex other than $s$ and $t$ is zero, and, thus, the net flow that leaves $s$ is equal to the net flow that enters $t$.
Therefore, each $st$-flow can be viewed as a circulation in another graph with an extra edge that connects $t$ to $s$.
Circulations are elements of the cycle space of the graph with coefficients taken over $\R$.
In a $d$-dimensional simplicial complex $\K$ the $d$-dimensional cycles are the formal sums (over $\R$) of $d$-dimensional simplices whose boundary is zero.
Because there are no $(d+1)$-simplices flows are the elements of the $d$th homology group $H_d(\K, \R)$.
The maximum flow problem in a simplicial complex asks to find an optimal element of $H_d(\K,\R)$ subject to capacity constraints.


The max-flow min-cut theorem states that in a graph the value of a maximum $st$-flow is equal to the value of a minimum $st$-cut.
This result is a special case of linear programming duality.
By rewriting the linear program in terms of the (co)boundary operator we obtain a similar result for simplicial complexes.
The question of whether or not a similar max-flow min-cut theorem holds for simplicial complexes was asked, and left open, in a paper by Latorre~\cite{latorre2012maxflow}.
We give a positive answer to this question, but with a caveat. 
When viewing flows and cuts from a topological point of view their linear programs are dual to one another.
However, we also provide a combinatorial definition of a cut which feels more natural for a minimization problem. 
Topological and combinatorial cuts are equivalent for graphs, but they become different in dimensions $d > 1$.
Flows in higher dimension, are dual to topological cuts, but not combinatorial cuts in general. 
From a computational complexity viewpoint the two notions of cuts are very different.
We show that computing a minimum topological cut can be solved via linear programming, but that computing a minimum combinatorial cut is \NP-hard.

A closely related problem is the problem of computing a max-flow in a graph which admits an embedding into some topological space. The most well-studied cases are planar graphs and the more general case when the graph embeds into a surface~\cite{Borradaile2009,cen-mcshc-09,HomFlow,en-mcsnc-11,f-faspp-87,Hassin1981MaximumFI,hj-oamfu-85,is-mfpn-79,insw-iamcmf-11,morell2020minimumcost,r-mstcp-83}.
Max-flows and min-cuts are computationally easier to solve in surface embedded graphs, especially planar graphs.
We consider this problem generalized to simplicial complexes.
Planar graphs are 1-dimensional complexes embedded in $\R^2$, in Section~\ref{sec:maxflow-embed} we consider the special case when a $d$-dimensional simplicial complex admits an embedding into $\R^{d+1}$.
These complexes naturally admit a dual graph which we use to compute maximum flows and minimum cuts (both topological and combinatorial).
We show that a maximum flow in a simplicial complex can be found by solving a shortest paths problem in its dual graph. This idea was used by Hassin to solve the maximum flow problem in planar graphs~\cite{Hassin1981MaximumFI}. 
Further, we show that finding a minimum topological cut can be done by finding a minimum cost circulation in its dual graph.
By setting the demand equal to one in the minimum cost circulation problem we obtain an algorithm computing a minimum combinatorial cut.

Maximum flows in graphs can be computed using the Ford-Fulkerson algorithm.
Moreover, the fact that the Ford-Fulkerson algorithm halts serves as a proof that there exists a maximum integral flow when the graph has integral capacity constraints.
In dimensions $d > 1$ the maximum flow may be fractional, even with integral capacity constraints. The problem arises due to the existence of torsion in simplicial complexes of dimension $d > 1$.
We show that despite the maximum flow being fractional the Ford-Fulkerson algorithm halts on simplicial complexes. However, in order for it to halt a special heuristic on picking the high dimensional analog of an augmenting path must be implemented.
Despite the algorithm halting it could we could not prove a polynomial upper bound on the number of iterations it takes.

\section{Preliminaries}\label{sec:prelim}
\paragraph{Simplicial homology}
We give a brief overview of some basic concepts from simplicial (co)homology that will be used throughout the paper. We recommend the book by Hatcher~\cite{Hatcher} for a more complete exposition.
In this paper $\K$ will always be a finite oriented $d$-dimensional simplicial complex which we now define. 
Given a finite set of vertices $V$ we define a $d$-dimensional simplex to be a subset of $d+1$ vertices of $V$.
$0$-simplices are vertices, $1$-simplices are edges, $2$-simplices are triangles, and so on.
We define an \textit{abstract simplicial complex} $\K$ to be a finite collection of simplices with the property that it is downward closed. That is, for every $\sigma \in \K$ if $\tau \subset \sigma$ then $\tau \in \K$. We call the subsets of a simplex the \textit{faces} of the simplex. The dimension of $\K$ is dimension of its largest simplex.
Further, we define an \textit{orientation} on $\K$ by fixing a linear ordering on the vertices in $V$ and treating simplices as ordered sets.
An oriented simplex is a permutation of the linear ordering, and the orientation of the simplex is the parity of the permutation. We call oriented simplices that agree with the linear ordering \textit{positively oriented} otherwise we call them \textit{negatively oriented}.

We define a $d$-chain to be a formal sum of $d$-simplices with coefficients over $\R$.
By $\K^d$ we denote the \textit{$d$-skeleton} of $\K$ which is the set of all $d$-simplices in $\K$ and we denote its cardinality by $n_d$.
Additionally, we may assign a weight function on the $d$-skeleton of $\K$ denoted $c \colon \K^d \rightarrow \R$. In the context of the maximum flow problem we will refer to the weight function as a \textit{capacity function}.

Given a simplicial complex $\K$ we define the $d$th \textit{chain space} $C_d(\K)$ to be the vector space generated by $\K^d$ with coefficients over $\R$.
Let $\sigma = [v_0,\dots,v_d]$ be a $d$-simplex with each $v_i$ a vertex.
We define the simplicial boundary operator $\partial_d \colon C_{d+1}(\K) \rightarrow C_d(\K)$ by \[\partial_d (\sigma) = \sum_{i=0}^d (-1)^i [v_0,\dots,\hat{v_i},\dots,v_d],\] where $[v_0,\dots,\hat{v_i},\dots,v_d]$ denotes the $(d-1)$-simplex obtained by removing $v_i$ from $\sigma$. 
The boundary operator extends linearly over $\R$ in the natural way.
The boundary operator is defined for all $0 \leq k \leq d$ and we will drop the subscript denoting it $\partial$ when the context is clear. Since we are assuming $C_d(\K)$ to be a finite dimensional vector space we can equivalently define the boundary operator to be the $n_{d-1} \times n_d$ matrix such that the entry $\partial_{i,j}$ is equal to $1$ or $-1$ (depending on the orientation) if the $i$th $(d-1)$-simplex is a face of the $j$th $d$-simplex, and 0 otherwise.
With the boundary operator we define the \textit{chain complex} on $\K$ to be the following sequence of vector spaces connected by their boundary operators.
\[\dots \xrightarrow{\partial_{d+1}} C_d(\K) \xrightarrow{\partial_d} C_{d-1}(\K) \xrightarrow{\partial_{d-1}} \dots \xrightarrow{\partial_2} C_1(\K) \xrightarrow{\partial_1} C_0(\K) \rightarrow 0 \]
The space of \textit{$d$-cycles} $Z_d(\K)$ is defined to be $\ker(\partial_d)$ and the space of \textit{$d$-boundaries} $B_d(\K)$ is defined to be $\im(\partial_{d+1})$. We call any $d$-cycle $\gamma \in B_d(\K)$ a \textit{null-homologous} cycle. Since $\partial_{d} \circ \partial_{d+1} = 0$ the quotient space $H_d(\K) = Z_d(\K) / B_d(\K)$ is well-defined, and we call $H_d(\K)$ the $d$th \textit{homology space} of $\K$.
By $\beta_d$ we denote the $d$th Betti number of $\K$ which is defined to be $\beta_d = \dim H_d(\K)$.

We obtain the \textit{cochain space} by dualizing $C_d(\K)$ in the following way.
The $d$th cochain space $C^d(\K)$ is defined to be the dual space of $C_d(\K)$, that is, the space of all linear functions $f \colon C_d(\K) \rightarrow \R$. We call $f$ a \textit{cochain} and we define the \textit{coboundary operator} $\delta_d \colon C^{d-1}(\K) \rightarrow C^{d}(\K)$ on cochains as the composition of functions $\delta_d f = f \circ \partial_d$. The cochain complex is the following sequence of cochain spaces.
\[\dots \xleftarrow{\delta_{d+1}} C^d(\K) \xleftarrow{\delta_d} C^{d-1}(\K) \xleftarrow{\delta_{d-1}} \dots \xleftarrow{\delta_2} C^1(\K) \xleftarrow{\delta_1} C^0(\K) \leftarrow 0 \]
The space of \textit{$d$-cocycles} $Z^d(\K)$ is defined to be $\ker(\delta_{d+1})$ and the space of \textit{$d$-coboundaries} $B^d(\K)$ is defined to be $\im(\delta_d)$. Again, since $\delta_{d+1} \circ \delta_d = 0$ the quotient space $H^d(\K) = Z^d(\K) / B^d(\K)$ is well-defined and we call $H^d(\K)$ the $d$th \textit{cohomology space} of $\K$.

Since $C_d(\K)$ and $C^d(\K)$ are finite dimensional vector spaces we have the isomorphisms $C_d(\K) \cong C^d(\K)$ and $H_d(\K) \cong H^d(\K)$. It can be shown that the coboundary operator as a matrix is equal to the transpose of the boundary operator: $\delta_d = \partial_d^T$.
We will view $d$-(co)chains as both $d$-dimensional vectors and as linear functions $C_d(\K) \rightarrow \R$ whenever it is convenient to do so. However, we will refer to flows as $d$-chains and cuts as $d$-cochains unless explicitly stated otherwise.
We will often want to talk about the underlying set of simplices of the (co)chain and refer to this set as the \textit{support} of the (co)chain; the support of a chain $\sigma = \sum \alpha_i \sigma_i$ is defined as the set $\supp(\sigma) = \{\sigma_i \in \K^d \mid \alpha_i \neq 0 \}$.

We will need to define the notion of \textit{relative homology} in order to cite known results about the boundary matrix of a simplicial complex.
Let $\K_0 \subseteq \K$ be a subcomplex of $\K$. We call the quotient space $C_d(\K, \K_0) = C_d(\K) / C_d(\K_0)$ the space of $d$-dimensional \textit{relative chains} and there is an induced mapping $\partial_d^{\K,\K_0} \colon C_d(\K, \K_0) \rightarrow C_{d-1}(\K, \K_0)$. From the induced mapping we define the spaces of \textit{relative $d$-cycles} $Z_d(\K, \K_0)$, \textit{relative $d$-boundaries} $B_d(\K, \K_0)$, and \textit{relative $d$-dimensional homology} $H_d(\K, \K_0)$ in the natural way.
Further, let $\mathcal{L} \subseteq \K$ be a \textit{pure} $d$-dimensional subcomplex; that is, every $(d-1)$-simplex in $\mathcal{L}$ is incident to some $d$-simplex in $\mathcal{L}$.
Let $\mathcal{L}_0 \subseteq \mathcal{L}$ be a pure $(d-1)$-dimensional subcomplex of $\mathcal{L}$.
The induced map on relative homology $\partial_d^{\mathcal{L},\mathcal{L}_0} \colon C_d(\mathcal{L}, \mathcal{L}_0) \rightarrow C_{d-1}(\mathcal{L},\mathcal{L}_0)$ has a natural matrix representation.
We construct the matrix $\partial_d^{\mathcal{L},\mathcal{L}_0}$ by starting with $\partial_d$ and including the columns corresponding to $d$-simplices in $\mathcal{L}$ while excluding the rows corresponding to $(d-1)$-simplices in $\mathcal{L}_0$.

\paragraph{Dual graphs}
We will provide algorithms for a special class of simplicial complexes: \textit{embedded complexes}. This class of simplicial complexes admits a natural graph structure known as the \textit{dual graph}.
Let $\K$ be a $d$-dimensional simplicial complex. 
We say $\K$ has an embedding into $\R^{d+1}$ if there exists a function $f \colon \K \rightarrow \R$ such that $f$ restricted to any $k$-simplex is an injection into a subspace homeomorphic to a $k$-dimensional disk. Moreover, for any two simplices $\sigma, \tau \in \K$ we require that $f(\sigma) \cap f(\tau) = f(\sigma \cap \tau)$. That is, the image of any two simplices only intersects at their shared faces.

Let $\K$ be a $d$-dimensional simplicial complex embedded in $\R^{d+1}$. The Alexander duality theorem implies that $\R^{d+1} \setminus \K$ consists of $\dim H_d(\K) + 1$ connected components~\cite{AlexanderDuality}, see ~\cite[Corollary 3.45]{Hatcher} for a modern treatment.
We define the dual graph $\K^*$ of $\K$ to be the graph whose vertices are in bijection with the connected components of $\R^{d+1} \setminus \K$ and whose edges are in bijection with the $d$-simplices of $\K$. Since simplices are two-sided each $d$-simplex is on the boundary of at most two connected components of $\R^{d+1} \setminus \K$, and the endpoints of its dual edge are dual to these connected components.
Exactly one connected component of $\R^{d+1} \setminus \K$ is unbounded and we denote the vertex dual to this component by $v_\infty$.

\paragraph{Integral coefficients}
In order to state our results regarding integral solutions we will need to define the integral homology groups of a simplicial complex. Given a simplicial complex $\K$ we define $C_d(\K, \Z), Z_d(\K, \Z), B_d(\K, \Z),$ and $H_d(\K, \Z)$ in the same was as before, but with coefficients over $\Z$ instead of $\R$. With this change we  no longer have vector spaces, but instead have free abelian groups.
The fundamental theorem of finitely generated abelian groups gives us the decomposition $H_d(\K, \Z) \cong \Z^k \oplus \Z_{t_1} \oplus \dots \oplus \Z_{t_n}$ for some $k \in \N$. We call the subgroup $\Z_{t_1} \oplus \dots \oplus \Z_{t_n}$ the \textit{torsion subgroup} of $H_d(\K, \Z)$ and when this subgroup is trivial we call the complex \textit{torsion-free}. We say $\K$ is relatively torsion-free in dimension $d$ if $H_d(\mathcal{L}, \mathcal{L}_0, \Z)$ is torsion-free for all subcomplexes $\mathcal{L}$ and $\mathcal{L}_0$. There exist complexes that are torsion-free but are not relatively torsion-free; see~\cite{Dey2011} for examples.

\paragraph{Total unimodularity}
Let $A$ be a matrix; we say that $A$ is \textit{totally unimodular} if every square submatrix $A'$ of $A$ has $\det(A') \in \{-1, 0, 1\}$.
Totally unimodular matrices are important in combinatorial optimization because linear programs with totally unimodular constraint matrices are guaranteed to have integral solutions~\cite{VeinottJr1968}.
Dey, Hirani, and Krishnamoorthy have provided topological conditions on when a simplicial complex has a totally unimodular boundary matrix~\cite{Dey2011} stated below.
We call a simplicial complex meeting the criteria of Theorem~\ref{thm:dey} \text{relative torsion-free} in dimension $d-1$.

\begin{theorem}[Dey et al.~\cite{Dey2011}, Theorem 5.2]\label{thm:dey}
Let $\K$ be a $d$-dimensional simplicial complex. The boundary matrix $\partial \colon C_d(\K) \rightarrow C_{d-1}(\K)$ is totally unimodular if and only if $H_{d-1}(\mathcal{L}, \mathcal{L}_0, \Z)$ is torsion-free for all pure subcomplexes $\mathcal{L}_0,\mathcal{L}$ of $\K$ of dimensions $d-1$ and $d$ where $\mathcal{L}_0 \subset \mathcal{L}$.
\end{theorem}


\paragraph{Hodge theory}
Throughout this paper we utilize discrete Hodge theory and recommend the survey by Lim \cite{hodge_survey} as an introduction to the topic.
In particular, we use the Hodge decomposition which can be stated as a result on real valued matrices satisfying $AB=0$.
\begin{theorem}[Hodge decomposition~\cite{hodge_survey}]\label{thm:hodge}
Let $A \in \R^{m \times n}$ and $B \in \R^{n \times p}$ be matrices satisfying $AB = 0$.
We can decompose $\R^n$ into the orthogonal direct sum, \[\R^n = \im(A^T) \oplus \ker(A^TA + BB^T) \oplus \im(B).\]
\end{theorem}
Setting $A = \partial_{d}$ and $B = \partial_{d+1}$ yields the Hodge decomposition for simplicial complexes. 
The middle term of the direct sum becomes $\ker(\delta_{d+1} \partial_{d+1} + \partial_d \delta_d)$. The linear operator $\delta_{d+1}\partial_{d+1} + \partial_d \delta_d$ is known as the combinatorial Laplacian of $\K$ which is a generalization of the graph Laplacian.  
Moreover, it can be shown that $\ker(\delta_{d+1} \partial_{d+1} + \partial_d \delta_d) \cong H_d(\K, \R)$.
We now state the Hodge decomposition on simplicial complexes as the following isomorphism:
\[C_d(\K, \R) \cong \im(\delta_d) \oplus H_d(\K, \R) \oplus \im(\partial_{d+1}).\]

\section{Flows and cuts}\label{sec:cutflow}
In this section we give an overview of our generalizations of flows and cuts from graphs to simplicial complexes.
Flows and cuts in higher dimensional settings have been studied previously. Duval, Klivans, and Martin have generalized cuts and flows to the setting of CW complexes~\cite{Duval2015}. Their definitions are algebraic; defining cuts to be elements of $\im(\delta)$ and flows to be elements of $\ker(\partial)$.
Our definitions are closely related, but are motivated by the algorithmic problems of computing maximum flows and minimum cuts.
In Section~\ref{subsec:top} we give definitions of flows and cuts from from the perspective of algebraic topology, and in Section~\ref{subsec:comb} we give a combinatorial definition of a cut in a simplicial complex.
The distinction between the two types of cuts will be important when formulating the minimum cut problem on simplicial complexes.

\subsection{Topological flows and cuts}\label{subsec:top}

First we briefly recall the definition of an $st$-flow in a directed graph $G=(V,E)$. An $st$-flow $f$ is a function $f \colon E \rightarrow \R$ satisfying the conservation of flow constraint: for all $v \in V \setminus \{s, t\}$ we have $\sum_{(u, v) \in E} f(u, v) = \sum_{(v, u) \in E} f(v, u)$. That is, the amount of flow entering the vertex equals the amount of flow leaving the vertex. 
The value of $f$ is equal to the amount of flow leaving $s$ (or equivalently, entering $t$).
Alternatively, we may view $f$ as a 1-chain and we have $\partial f = k(t - s)$ where $k$ is the value of $f$.
Note that $t - s$ is a null-homologous 0-cycle. 
More generally, for any null-homologous $(d-1)$-cycle $\gamma$ we call a $d$-chain $f$ with $\partial f = k \gamma$ a $\gamma$-flow of value $k$.
Note that under our naming convention an \say{$st$-flow} in a graph would be called a $(t-s)$-flow. However, in the case of graphs we use the traditional naming convention and call a flow from $s$ to $t$ an $st$-flow.

\begin{definition}[$\gamma$-flow]
Let $\K$ be a $d$-dimensional simplicial complex and $\gamma$ be a null-homologous $(d-1)$-cycle in $\K$. A \textbf{$\gamma$-flow} is a $d$-chain $f$ with $\partial f = k \gamma$ where $k \in \R$. We call $k$ the value of the flow $f$ and denote the value of $f$ by $\|f\|$.
We say that $f$ is feasible with respect to a capacity function $c \colon \K^d \rightarrow \R^+$ if $0 \leq f(\sigma) \leq c(\sigma)$ for all $\sigma \in \K^d$.
\end{definition}

Our definition of a $\gamma$-flow is very close to the algebraic definition which is element of $\ker(\partial)$. Given a simplicial complex $\K$ and a $\gamma$-flow $f$ of value $k$ we convert $f$ into a \textit{circulation}, where a circulation is defined to be an element of $\ker(\partial)$.
To convert $f$ into a circulation we add an additional basis element to $C_d(\K)$, call it $\Sigma$, whose boundary is $\partial \Sigma = -\gamma$. This operation is purely algebraic; we should think of it as operating on the chain complex rather than the underlying topological space. 
Now we construct the circulation $f' = f + k\Sigma$. We call any circulation built from a $\gamma$-flow a $\gamma$-circulation. Clearly, $f' \in \ker(\partial)$ in the new chain complex. Moreover, there is a clear bijection between $\gamma$-flows and $\gamma$-circulations.
The value of the circulation is the value of $f'(\Sigma)$, so this bijection preserves the value.

We now shift our focus to the generalization of cuts to a simplicial complex.
The algebraic definition, elements of $\im(\delta)$, is natural. The cut space of a graph is commonly defined to be the space spanned by the coboundaries of each vertex.
In a simplicial complex $\K$, removing the support of a $d$-chain in $\im(\delta)$ increases $\dim H_{d-1}(\K)$.
In a graph $G$, removing the support of any 1-chain in $\im(\delta)$ increases $\dim H_0(G)$ which is equivalent to increasing the number of connected components of $G$.

The above definition implies that a cut is a $d$-chain in a $d$-dimensional simplicial complex. However, for our purposes we will define a cut to be a $(d-1)$-cochain.
To motivate our definition we recall the notion of an $st$-cut in a graph.
An $st$-cut in a graph is a partition of the vertices into sets $S$ and $T$ such that $s \in S$ and $t \in T$. Define $p \colon V(G) \rightarrow \{0, 1\}$ such that $p(v) = 1$ if $v \in S$ and $p(v) = 0$ if $v \in T$. The support of the coboundary of $p$ is a set of edges whose removal destroys all $st$-paths. 
That is, upon removing the support, the 0-cycle $t - s$ is no longer null-homologous.
Moreover, $p$ is a $0$-cochain with $p(t - s) = -1$. The sign of $p(t - s)$ will be important when we consider directed cuts. With this in mind we define our notion of a $\gamma$-cut.

\begin{definition}[$\gamma$-cut]
Let $\K$ be a $d$-dimensional simplicial complex with weight function $c \colon \K^d \rightarrow \R^+$ and $\gamma$ be a null-homologous $(d-1)$-cycle in $\K$. A \textbf{$\gamma$-cut} is a $(d-1)$-cochain $p$ such that $p(\gamma) = -1$.
Denote the coboundary of $p$ as the formal sum $\delta(p) = \sum \alpha_i \sigma_i$, we define the size of a $\gamma$-cut to be $\|p\| = \sum |\alpha_i c(\sigma_i)|$.
\end{definition}

Because of the requirement that $p(\gamma) = -1$ we call $p$ a \textit{unit $\gamma$-cut}.
By relaxing this requirement to $p(\gamma) < 0$ the cochain $p$ still behaves as a $\gamma$-cut, but its size can become arbitrarily small by multiplying by some small value $\epsilon > 0$.
We justify our definition with the following proposition which shows that removing the support of the coboundary of a $\gamma$-cut prevents $\gamma$ from being null-homologous.

\begin{proposition}\label{prop:cut}
Let $\K$ be $d$-dimensional simplicial complex and $p$ be a $\gamma$-cut. The cycle $\gamma$ is not null-homologous in the subcomplex $\K \setminus \supp(\delta(p))$.
\end{proposition}
\begin{proof}
By way of contradiction let $\Gamma$ be a $d$-chain in $\K \setminus \supp(\delta(p))$ such that $\partial \Gamma = \gamma$. 
Since $\delta(p) = 0$ in $\K \setminus \supp(\delta(p))$, we have that $\langle \Gamma, \delta(p) \rangle = 0$ in $\K \setminus \supp(\delta(p))$. However, this implies that $0 = \langle \Gamma, \delta(p) \rangle = \langle \partial \Gamma, p \rangle = \langle \gamma, p \rangle = p(\gamma) = 0$, a contradiction as $p$ is a $\gamma$-cut and $p(\gamma) = -1$. 
\end{proof}

\subsection{Combinatorial cuts}\label{subsec:comb}
Alternatively, we can view a $\gamma$-cut as a discrete set of $d$-simplices rather than a $d$-chain. 
In the case of graphs a combinatorial $st$-cut is just a set of edges whose removal disconnects $s$ from $t$.
This distinction will become important when we consider the minimization problem of finding a minimum cost set of $d$-simplices whose removal prevents $\gamma$ from being null-homologous.

\begin{definition}[Combinatorial $\gamma$-cut]
Let $\K$ be a $d$-dimensional simplicial complex with weight function $c \colon \K^d \rightarrow \R^+$ and $\gamma$ be a null-homologous $(d-1)$-cycle in $\K$. A \textbf{combinatorial $\gamma$-cut} is a set of $d$-simplices $C \subseteq \K^d$ such that $\gamma$ is not null-homologous in $\K \setminus \supp(C)$.
The size of a combinatorial $\gamma$-cut is defined by the sum of the weights of the $d$-simplices $\|C\| = \sum_{\sigma \in C} c(\sigma)$.
\end{definition}

The next proposition shows a relationship between $\gamma$-cuts and combinatorial $\gamma$-cuts. Removing a combinatorial $\gamma$-cut $C$ from $\K$ increases $\dim H_{d-1}(\K)$. 
This is because removing $C$ must decrease the rank of $\partial_d$ and by duality this also decreases the rank of $\delta_d$ which increases the dimension of $H^{d-1}(\K) \cong H_{d-1}(\K)$.
It follows that $C$ must contain the support of some coboundary. Given an additional minimality condition on $C$ we show that $C$ is equal to the support of some coboundary.

\begin{proposition}
Let $C$ be a combinatorial $\gamma$-cut in a $d$-dimensional simplicial complex $\K$. 
Further, assume that $C$ is minimal in the sense that for all $C' \subset C$ we have $\dim H_{d-1}(\K \setminus C') < \dim H_{d-1}(\K \setminus C)$.
There exists a $(d-1)$-cochain $p$ such that $\supp(\delta(p)) = C$.
\end{proposition}
\begin{proof}
Recall that $H_{d-1}(\K) = \ker(\partial_{d-1})/\im(\partial_d)$, $H^{d-1} = \ker(\delta_d)/\im(\delta_{d-1})$, and over real coefficients we have the isomorphism $H_{d-1}(\K) \cong H^{d-1}(\K)$.
Removing the set of $d$-simplices $C$ from $\K$ does not affect $\ker(\partial_{d-1})$ or $\im(\delta_{d-1})$.
However, as $C$ is a combinatorial $\gamma$-cut it must decrease the dimension of $\im(\partial_d)$ and increase the dimension of $\ker(\delta_d)$.
It follows that there must exist some $(d-1)$-cochain $p$ such that $\delta_d(p) \neq 0$ in $\K$ but $\delta_d(p) = 0$ in $\K \setminus C$. 
Hence, $\supp(\delta_d(p)) \subseteq C$.
Without loss of generality we choose $p$ such that $|\supp(\delta_d(p))|$ is maximized.
Define $C' \coloneqq C \setminus |\supp(\delta_d(p))|$.
By minimality we have $\dim H_{d-1}(\K \setminus C') < \dim H_{d-1}(\K \setminus C)$, so there is some $(d-1)$-cycle $\gamma'$ that is null-homologous in $\K \setminus C'$ but not null-homologous in $\K \setminus C$.
Dually, there is some $(d-1)$-cochain $p'$ with $\delta_d(p') \neq 0$ in $\K \setminus C'$ but $\delta_p(p') = 0$ in $\K \setminus C$.
It follows that $\delta_d(p + p') = 0$ in $\K \setminus C$ but $\delta_d(p + p') \neq 0$ in $\K$, so $\supp(p + p') \subseteq C$.
Since $\supp(p)$ and $\supp(p')$ are disjoint we have $|\supp(p)| < |\supp(p + p')|$, contradicting our assumption that $p$ maximizes $|\supp(p)|$.
\end{proof}

In graphs the linear program solving the minimum $st$-cut problem takes as input a directed graph and returns a set of directed edges whose removal destroys all directed $st$-paths. This is called a directed cut. After removing the directed cut the 0-cycle $t-s$ may still be null-homologous; we can find a 1-chain with boundary $t-s$ using negative coefficients to traverse an edge in the backwards direction.
In order to generalize the minimum cut linear program to simplicial complexes we will need to define a directed combinatorial $\gamma$-cut, which requires the additional assumption that the $d$-simplices of $\K$ are oriented.

\begin{definition}[Directed combinatorial $\gamma$-cut]
Let $\K$ be an oriented $d$-dimensional simplicial complex with weight function $c \colon \K^d \rightarrow \R$ and $\gamma$ be a null-homologous $(d-1)$-cycle in $\K$. A \textbf{directed combinatorial $\gamma$-cut} is a set of $d$-simplices $C \subset \K^d$ such that in $\K \setminus \supp(C)$ there exists no $d$-chain $\Gamma$ with non-negative coefficients such that $\partial\Gamma = \gamma$.
The size of a directed combinatorial $\gamma$-cut is defined by the sum of the weights of the $d$-simplices $\|C\| = \sum_{\sigma \in C} c(\sigma)$.
\end{definition}

Given a directed graph consider an $st$-cut given by the cochain definition. That is, a 0-cochain $p \colon V(G) \rightarrow \{0, 1\}$ with $p(s) = 1$ and $p(t) = 0$ partitioning $V$ into $S$ and $T$.
The support of $\delta(p)$ consists of two types of edges: edges leaving $S$ and entering $T$, and edges leaving $T$ and entering $S$. If $e \in E$ leaves $S$ and enters $T$ we have $(p \circ \partial)(e) = -1$ and if $e$ leaves $T$ and enters $S$ we have $(p \circ \partial)(e) = 1$. To construct a directed $st$-cut we simply take all of the edges mapped to $-1$.
The following proposition shows that we can build a directed combinatorial $\gamma$-cut from a coboundary just like in the case of directed graphs.

\begin{proposition}\label{prop:dircomcut}
Let $p$ be a $\gamma$-cut with coboundary $\delta(p) = \sum \alpha_i \sigma_i$. The set of $d$-simplices $C = \{ \sigma_i \mid \alpha_i < 0 \}$ is a directed combinatorial $\gamma$-cut.
\end{proposition}
\begin{proof}
By way of contradiction let $\Gamma$ be a non-negative $d$-chain in $\K \setminus C$ with $\partial \Gamma = \gamma$.
By the definition of $C$ we have $\langle \Gamma, \delta(p) \rangle \geq 0$ in $\K \setminus C$.
Construct a new chain complex by adding an additional basis element $\Sigma$ to $C_d(\K)$ such that $\partial \Sigma = - \gamma$. By construction $\langle \Sigma, \delta(p) \rangle = -p(\gamma) = 1$, hence we have $\langle \Gamma + \Sigma, \delta(p) \rangle > 0$.
However, $\Gamma + \Sigma$ is a $d$-cycle and $\delta(p)$ is a $d$-coboundary, so the Hodge decomposition ensures that they are orthogonal. Hence, $\langle \Gamma + \Sigma, \delta(p) \rangle = 0$, a contradiction.
\end{proof}

To conclude the section we will show that computing a minimum combinatorial $\gamma$-cut is \NP-hard. As we will see in Section~\ref{sec:linprog} minimum topological $\gamma$-cuts can be computed with linear programming.
Our hardness result holds for both the directed and undirected cases.
Our hardness result is a reduction from the well-known \NP-hard hitting set problem which we will now define.
Given a set $S$ and a collection of subsets $\Sigma = (S_1,\dots,S_n)$ where $S_i \subseteq S$ the hitting set problem asks to find the smallest subset $S' \subseteq S$ such that $S' \cap S_i \neq \emptyset$ for all $S_i$. We call such a subset $S'$ a \textit{hitting set} for $(S, \Sigma)$.

\begin{theorem}\label{thm:combcuthard}
Let $\K$ be a $d$-dimensional simplicial complex and $\gamma$ be a null-homologous $(d-1)$-cycle. Computing a minimum combinatorial $\gamma$-cut is \NP-hard for $d \geq 2$.
\end{theorem}
\begin{proof}
Our proof is a reduction from the hitting set problem.
First we consider the case when $d=2$ then we generalize to any $d \geq 2$.
Let $S$ be a set and $\Sigma = (S_1,\dots,S_n)$ where each $S_i \subseteq S$.
We construct a $2$-dimensional simplicial complex $\K$ from $S$ and $\Sigma$ in the following way.
For each $S_i \in \Sigma$ construct a triangulated disk $\D_i$ such that $\partial \D_i = \gamma$. That is, each $\D_i$ shares the common boundary $\gamma$.
To accomplish this we construct each $\D_i$ by beginning with a single triangle $t$ with $\partial t = \gamma$ and repeatedly adding a new vertex in the center of some triangle with edges connecting it to every vertex in that triangle. 
By this process we can construct a disk containing any odd number of triangles as each step increments the number of triangles in the disk by two. 
Moreover, at each step the boundary of the disk is always $\gamma$.
We construct each disk $\D_i$ such that $\D_i$ consists of one triangle $t_{i,s}$ for each element $s \in S_i$ and potentially one extra triangle $t'_i$ in the case that $|S_i|$ is even.
Next, for each $s \in S$ and $S_i$ with $s \in S_i$, we construct the quotient space by identifying each $t_{i,s}$ into a single triangle.
A minimum combinatorial $\gamma$-cut $C$ must contain at least one triangle from each $\D_i$ and without loss of generality we can assume $C$ does not contain any $t'_i$. 
If $t_i' \in C$ then by minimality it is the only triangle in $C \cap \supp(\D_i)$ and we can swap it with any other triangle in $\D_i$ without increasing the size of the cut.
By construction $C$ is a hitting set for $(S, \Sigma)$ since each $C \cap \supp(\D_i) \neq \emptyset$ for all $\D_i$.

To perform the above construction in higher dimensions we simply start with a single $d$-simplex $\sigma = [v_1,\dots,v_{d+1}]$ with boundary $\partial\sigma = \gamma$.
We subdivide $\sigma$ in the following way: add an additional vertex $v_{d+2}$ and replace $\sigma$ with the $d$-simplices $\sigma_i := [v_1,\dots,v_{i-1},v_{d+2},v_{i+1},\dots,v_{d+1}]$ for each $1 \leq i \leq d+1$. 
At each step we increase the number of $d$-simplices by $d$; moreover, at each step the complex remains homeomorphic to a $d$-dimensional disk. 
Our final complex will have at most $d$ extra $d$-simplices so for any fixed dimension $d$ the size of the complex is within a constant factor of the size of the given hitting set instance.
It remains to show that the subdivision process does not change the boundary of the disk.
To accomplish this we will show that $\partial \sigma = \sum_{i=1}^{d+1} \partial \sigma_i$.

For each $\sigma_i$ the boundary $\partial \sigma_i$ has $d$ terms in its sum; each term is a $(d-1)$-simplex which consists of $d$ vertices. Consider the matrix $A$ such that the entry $A_{i,j}$ contains the $j$th term of $\partial \sigma_i$. Note that no term in the $j$th column will contain the vertex $v_j$. Also note that the sum of the diagonal of $A$ is equal to $\partial \sigma$.
Below is an example for $d=3$.
\begin{equation*}
    \begin{bmatrix}
        v_2 v_3 v_4 & -v_5 v_3 v_4 & v_5 v_2 v_4 & -v_5 v_2 v_3\\
        v_5 v_3 v_4 & -v_1 v_3 v_4 & v_1 v_5 v_4 & -v_1 v_5 v_3\\
        v_2 v_5 v_4 & -v_1 v_5 v_4 & v_1 v_2 v_4 & -v_1 v_2 v_5\\
        v_2 v_3 v_5 & -v_1 v_3 v_5 & v_1 v_2 v_5 & -v_1 v_2 v_3
    \end{bmatrix}
\end{equation*}
The matrix $A$ is almost symmetric. The entries $A_{i,j}$ and $A_{j,i}$ contain the same vertices but possibly differ up to a sign or a permutation of the vertices. We want to show that $A_{i,j} = -A_{j,i}$ so that the sum $\sum_{i=1}^{d+1} \partial \sigma_i$ only contains diagonal which is equal to $\partial \sigma$.
The ordering of vertices in $A_{i,j}$ differs from $A_{j,i}$ by the placement of $v_{d+2}$, which we now characterize. Let $v_{i,j}$ denote the vertex in $\sigma_i$ that is not included in the term $A_{i,j}$.
Without loss of generality assume $A_{i,j}$ is in the upper triangle. Then $v_{d+2}$ is in the $i$th position of $A_{i,j}$ because $v_{i,j}$ appears after it in the ordering of $\sigma_i$.
It follows that $v_{d+2}$ is in the $(j-1)$th in the term $A_{j,i}$ since $v_{j,i}$ must appear before $v_{d+2}$ in the ordering of $\sigma_j$. So, the position of $v_{d+2}$ in $A_{i,j}$ and $A_{j,i}$ differs by $|i - j + 1|$ and this is the number of transpositions needed to permute $A_{i,j}$ into $A_{j,i}$.
When $i \equiv j \mod 2$ the terms $A_{i,j}$ and $A_{j,i}$ have the same sign, but differ by an odd permutation so $A_{i,j} = -A_{j,i}$. Similarly, when $i \not\equiv j \mod 2$ the terms $A_{i,j}$ and $A_{j,i}$ have opposite signs, but differ by an even permutation so $A_{i,j} = -A_{j,i}$ which concludes the proof.
\end{proof}

\section{Linear programming}\label{sec:linprog}
\subsection{Max-flow min-cut}
A \emph{simplicial flow network} is a tuple $(\K, c, \gamma)$ where $\K$ is an oriented $d$-dimensional simplicial complex, $c$ is the \emph{capacity function} which is a non-negative function $c \colon \K^d \rightarrow \R^{\geq 0}$, and $\gamma$ is a null-homologous $(d-1)$-cycle.
In a simplicial flow network we work with real coefficients; that is, we consider the chain groups $C_k(\K, \R)$.
In order to utilize the Hodge decomposition (Theorem~\ref{thm:hodge}) we modify $C_d(\K)$ by adding an additional basis element $\Sigma$ such that $\partial \Sigma = - \gamma$. Moreover, we extend the capacity function such that $c(\Sigma) = \infty$.
This allows us to work with circulations instead of flows while leaving the solution unchanged. The notation $n_d$ will refer to the number of basis elements in $C_d(\K, \R)$ which is now one more than the number of $d$-simplices in the underlying simplicial complex.

The goal of the maximum flow problem is to find a $d$-chain $f$ obeying the capacity constraints such that $\partial f = k \gamma$ where $k \in \R$ is maximized. 
Equivalently, we find a $d$-cycle $f$ which maximizes $f(\Sigma)$.
The linear program for the max-flow problem in a simplicial flow network is identical to the familiar linear program for graphs, but expressed in terms of the coboundary operator.
In a graph, conservation of flow at a vertex $v$ is the constraint $\delta_1(v) \cdot f = 0$; to formulate the linear program in higher dimensions we simply replace vertices with $(d-1)$-simplices.
The Hodge decomposition states that cycles are orthogonal to coboundaries, so conservation of flow ensures that $f$ is indeed a cycle. We now state the linear program for max-flow in a simplicial flow network.

\begin{equation}\label{lp:flow}\tag{LP1}
\begin{aligned}
\text{maximize}& \quad f(\Sigma)\\
\text{subject to}& \quad \delta(\tau) \cdot f = 0 &&\text{for each $\tau \in \K^{d-1}$}\\
&\quad 0 \leq f(\sigma) \leq c(\sigma) &&\text{for each $\sigma \in \K^d$}
\end{aligned}
\end{equation}

We dualize \ref{lp:flow} to obtain a generalization of the minimum cut problem in directed graphs. To make the dualization more explicit we will write out \ref{lp:flow} in matrix form: maximize $s \cdot f$ subject to $Af \leq b$ and $f \geq 0$, where we have \[A = \begin{bmatrix} \partial \\ -\partial \\ I_{n_d} \end{bmatrix},\, b = \begin{bmatrix} 0_{n_{d-1}} \\ 0_{n_{d-1}} \\ c \end{bmatrix},\, s = \begin{bmatrix} 0_{n_d - 1} \\ 1 \end{bmatrix}.\]
The matrix $A$ has dimension $(2 n_{d-1} + n_d) \times n_d$.
In our notation $I_k$ is the $k \times k$ identity matrix and $0_k$ is the $k \times 1$ column vector consisting of all zeros. Since the value of the flow is equal to $f(\Sigma)$ the vector $s$ is all zeros except for the final entry which is indexed by $\Sigma$ and receives an entry equal to one.
Further, $c$ is the $n_d \times 1$ capacity vector indexed by the $d$-simplices such that the entry indexed by $\sigma$ has value equal to $c(\sigma)$.

We can now state the dual program in matrix form: minimize $y \cdot b$ subject to $y^T A \geq s$ and $y \geq 0$. The vector $y$ is a $(2 n_{d-1} + n_d) \times 1$ column vector indexed by both the $(d-1)$-simplices and the $d$-simplices. However, only the entries indexed by $d$-simplices contribute to the objective function since $b$ is zero everywhere outside of the capacity constraints. We will denote the truncated vector consisting of entries indexed by $d$-simplices by $y_{d}$ and the entry corresponding to the $d$-simplex $\sigma \in \K^d$ will be denoted by $y_d(\sigma)$. Similarly we have two truncated vectors $y^1_{d-1}$ and $y^2_{d-1}$ corresponding to the entries indexed by the $(d-1)$-simplices.
Moreover, the rows of $y^T A \geq s$ are in the form
\[(y_{d-1}^1 - y_{d-1}^2)^T \partial + y_d \geq s.\]
For simplicity we define $y_{d-1} = y_{d-1}^1 - y_{d-1}^2$ and write $y_{d-1}(\tau)$ for the entry indexed by the $(d-1)$-simplex $\tau$.
Putting this together, we state the dual linear program as follows.

\begin{equation}\label{lp:cut}\tag{LP2}
\begin{aligned}
\text{minimize}& \quad \sum_{\sigma \in \K^d} y_d(\sigma) c(\sigma) \\
\text{subject to}& \quad y_{d-1} \cdot \partial \sigma + y_d(\sigma) \geq 0 &&\text{for each $\sigma \in \K^d$}\\
& \quad  y_{d-1} \cdot \partial \Sigma + y_d(\Sigma) = 1 \\
& \quad y_d \geq 0
\end{aligned}
\end{equation}
Note the strict equality in the second constraint does not follow from the duality.
However, we can assume a strict equality since if $y_{d-1} \cdot \partial \Sigma + y_d(\Sigma) > 1$ we can multiply $[y_{d-1},y_d]^T$ by some scalar $\epsilon < 1$ to make the inequality tight. This multiplication only decreases the value of $\sum y_d(\sigma)c(\sigma)$ so it does not change the optimal solution.

In the case of graphs \ref{lp:cut} has dual variables for vertices and edges. 
Moreover, there exists an integral solution such that each vertex is either assigned a 0 or a 1 since a graph cut is a partition of the vertices. 
The second inequality requires $y_0(s) = 1$ and $y_0(t) = 0$. 
To see this, when solving an $st$-cut on a graph, the basis element $\Sigma$ is an edge with $\partial\Sigma = s - t$, and $y_1(\Sigma) = 0$ otherwise the solution is infinite.
This naturally defines a partition of the vertices: $S$ containing vertices assigned a 1, and $T$ containing vertices assigned a 0. 
The constraints force an edge to be assigned a 1 if it leaves $S$ and enters $T$, otherwise it is assigned a 0. 
This solution can be interpreted as a 0-cochain $p$ with $p(st) = 1$, or in the notation of our definition of a simplicial cut: $p(t - s) = -1$.
Further, $y_1(e) = 1$ for every edge $e$ that is negative on $\delta(p)$ and a 0 otherwise, hence $y_1$ fits our definition of a directed $st$-cut in a 1-complex.
We will show the same result holds in higher dimensions; that is, $y_d$ is a directed $\gamma$-cut arising from the $(d-1)$-cochain $y_{d-1}$.

\begin{lemma}\label{lem:feasible1}
Let $y = [y_{d-1},y_d]^T$ be an optimal solution to \ref{lp:cut}. The set $\supp(y_d)$ is a directed combinatorial $\gamma$-cut.
\end{lemma}
\begin{proof}
Note that $y_{d-1}$ can be interpreted as a $(d-1)$-cochain. 
Since $c(\Sigma) = \infty$ we have $y_d(\Sigma) = 0$, otherwise the solution is infinite.
It follows from the second constraint that $-y_{d-1}(\gamma) = y_{d-1} (-\gamma) = y_{d-1} \cdot \partial \Sigma = 1$, hence $y_{d-1}(\gamma) = -1$ making $y_{d-1}$ a $\gamma$-cut.
Expanding $\delta(y_{d-1})$ into a linear combination of $d$-simplices $\sum \alpha_i \sigma_i$ and applying the first inequality constraint gives us the set equality $\supp(y_d) = \{ \sigma_i \mid \alpha_i < 0 \}$ since $y_d(\sigma_i) = 0$ precisely when $\alpha_i > 0$.
Thus, the result follows from \autoref{prop:dircomcut}.
\end{proof}

\begin{lemma}\label{lem:feasible2}
Let $p$ be a $\gamma$-cut with coboundary $\delta(p) = \sum \alpha_i \sigma_i$ and let $\delta(p)^- = \sum_{\alpha_i < 0} \alpha_i \sigma_i$. The vector $[p, -\delta(p)^-]^T$ is a finite feasible solution to \ref{lp:cut}.
\end{lemma}
\begin{proof}
We view the (co)chains $p$ and $-\delta(p)^-$ as vectors and let them take the roles of $y_{d-1}$ and $y_d$, respectively.
The first constraint is satisfied since for all $\sigma \in \K^d$ we have $-\delta(p)^-(\sigma) = - p \cdot \partial(\sigma)$ when $p \cdot \partial(\sigma)$ is negative, and $-\delta(p)^-(\sigma)=0$ otherwise. The second constraint and the finiteness of the solution is satisfied by the fact that $p(\gamma) = -1$.
\end{proof}

Lemma~\ref{lem:feasible1} tells us that a solution to \ref{lp:cut} yields a directed combinatorial $\gamma$-cut.
Recall, by \autoref{prop:dircomcut} every $\gamma$-cut $p$ yields a directed combinatorial $\gamma$-cut by taking the coboundary $\delta(p) = \sum \alpha_i \sigma_i$ and considering the negative components $\delta(p)^- = \{\sigma_i \mid \alpha_i < 0\}$. 
By Lemma~\ref{lem:feasible2} $\delta(p)^-$ is a feasible solution to \ref{lp:cut}; the cost of this solution is $c \cdot \delta(p)^-$.
The coefficients $\alpha_i$ need not always equal one; hence in general we have $\|C\| \neq c \cdot \delta(p)^-$.
It follows that \ref{lp:cut} need not return a minimum directed combinatorial $\gamma$-cut.
In Theorem~\ref{thm:totuni} we will give conditions describing when \ref{lp:cut} returns a directed combinatorial $\gamma$-cut.
To conclude the section we state our main theorem about \ref{lp:cut} whose proof is immediate from Lemmas \ref{lem:feasible1} and \ref{lem:feasible2}.

\begin{theorem}\label{thm:maxflowmincut}
Let $y = [y_{d-1}, y_d]^T$ be an optimal solution to \ref{lp:cut}. 
The set $\supp(y_d)$ is a directed combinatorial $\gamma$-cut such that $y_d = \delta(y_{d-1})^-$.
Moreover, $y_d$ minimizes $c \cdot \delta(p)^-$ where $p$ ranges over all $\gamma$-cuts.
\end{theorem}

\subsection{Integral solutions}\label{sec:integral_flow}
In this section we provide an example of a simplicial flow network with integral capacity constraints and fractional maximum flow. By Theorem~\ref{thm:dey} such a network must contain some relative torsion. This is achieved by the inclusion of a \mobius{} strip in our simplicial flow network. Our example will be used later in Theorem~\ref{thm:inthard} showing that computing a maximum integral flow in a simplicial flow network is \NP-hard. 

\begin{figure}[!htb]
   \centering
    \includegraphics[scale=.5]{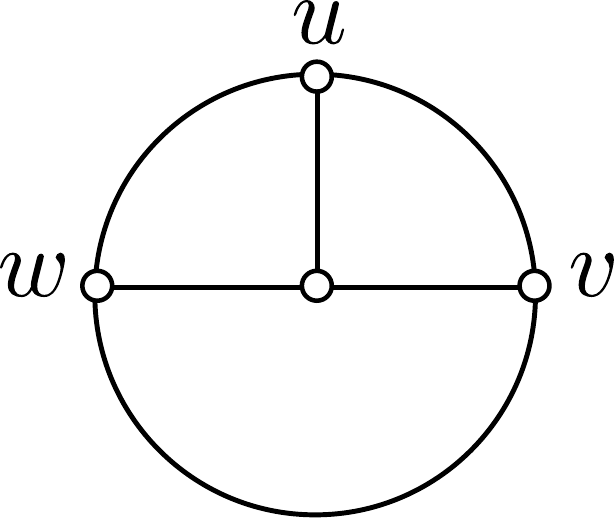}\qquad
    \includegraphics[scale=.5]{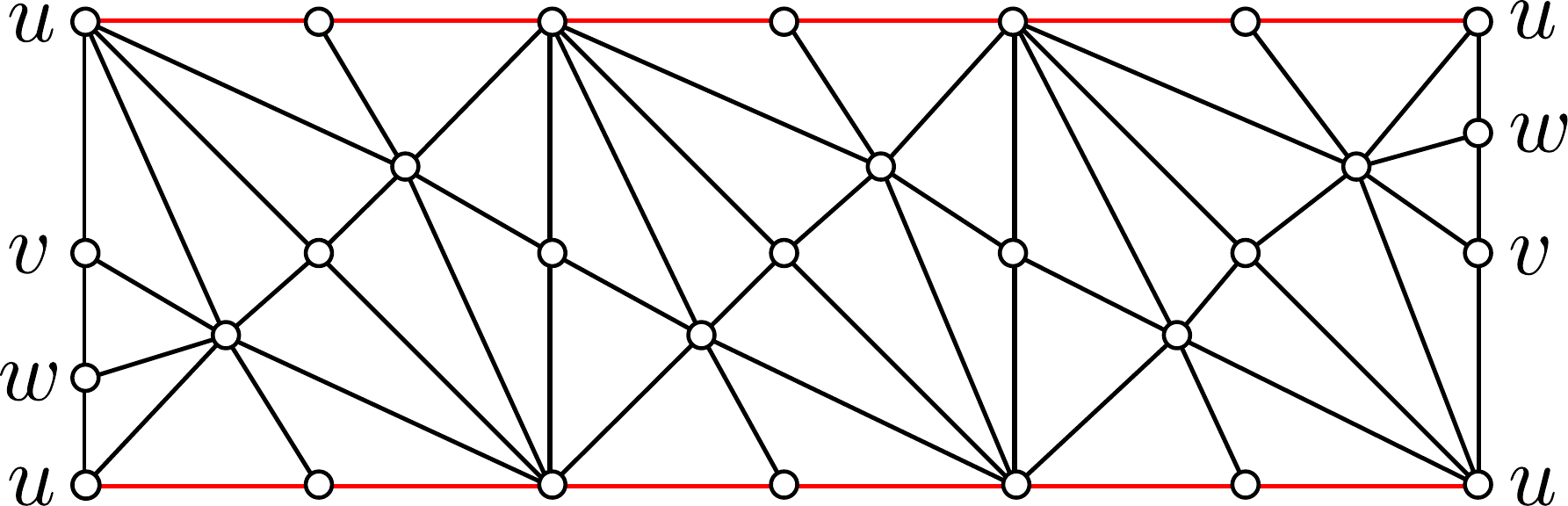}
 \caption{A triangulated disk $\mathcal{D}$ (left) and \mobius{} strip $\M$ (right). The \mobius{} strip has two points on its boundary identified forming the vertex $u$. In red we have the input cycle (a figure-eight) $\gamma$ and we set $\alpha = uwvu$. We orient the complex such that $\partial \mathcal{M} = \gamma + 2\alpha$ and  $\partial \mathcal{D} = -\alpha$. The capacity on each simplex in both the disk and \mobius{} strip is one.}
 \label{fig:mbdisk}
\end{figure}

\paragraph{Fractional maximum flow}
We will now explicitly describe a simplicial flow network with integral capacities whose maximum flow value is fractional.
Let $\M$ be a triangulated \mobius{} strip with boundary $\partial \M = 2\alpha + \gamma$ such that two vertices in $\alpha$ have been identified making $\alpha$ a simple cycle. 
This identification makes $\gamma$ a figure-eight.
Now let $\D$ be a triangulated disk oriented such that $\partial \D = - \alpha$. See Figure~\ref{fig:mbdisk} for an illustration. Call the resulting complex $\M\D$.
The capacity function $c$ has $c(t) = 1$ for each triangle $t \in \M\D$. Now we solve the max-flow problem on $(\M\D, c, \gamma)$. Note that for any flow $f$ we have $f(t_1) = f(t_2)$ for all triangles $t_1,t_2 \in \M$; moreover, for all $t_1,t_2 \in \D$ we also have $f(t_1)=f(t_2)$.
The value of any flow $f$ on $(\M\D,c,\gamma)$ is equal to its value on $\M$, and in order to maintain conservation of flow we must have $f(\D) = 2f(\M)$.
Now, the capacity constraints imply that the maximum flow $f$ has $f(\M) = 1/2$ and $f(\D) = 1$ . We have $\partial f = \frac{1}{2}\partial \M + \partial \D = \frac{1}{2}\gamma + \alpha - \alpha$. Hence, $\|f\| = 1/2$.

\paragraph{Maximum integral flow}
Given a simplicial flow network $(\K, c, \gamma)$ with integral capacities we consider the problem of finding the maximum integral flow. That is, a $d$-chain $f \in C_d(\K, \Z)$ obeying the capacity constraints such that $\partial f = k\gamma$ where $k \in \Z$ is maximized.
We show the problem is \NP-hard by a reduction from graph 3-coloring. Our reduction is inspired by a MathOverflow post from Sergei Ivanov showing that finding a subcomplex homeomorphic to the 2-sphere is \NP-hard~\cite{MOhardness}. In the appendix we adapt the proof to show that the high dimensional generalization of computing a directed path in a graph is also \NP-hard.
Given a graph $G$ we construct a 2-dimensional simplicial flow network whose maximum flow is integral if and only if $G$ is 3-colorable.

\begin{theorem}\label{thm:inthard}
Let $(\K, c, \gamma)$ be a simplicial flow network where $\K$ is a 2-complex and $c$ is integral. Computing a maximum integral flow of $(\K, c, \gamma)$ is \NP-hard.
\end{theorem}
\begin{proof}
Let $G = (V, E)$ be a graph. We will construct a simplicial flow network $(\K, c, \gamma)$ such that its maximum flow is integral if and only if $G$ is 3-colorable.

We start our construction with a punctured sphere $\mathcal{S}$ containing $|V| + 1$ boundary components called $\gamma$ and $\beta_v$ for each $v \in V$. For each boundary component $\beta_v$ we construct three disks $R_v,B_v,G_v$ each with boundary $-\beta_v$.
These disks represent the three colors in our coloring: red, blue, and green. 
We refer to these disks as \textit{color disks} and use $\C_v$ to denote an arbitrary color disk associated with $v$ and use $k \in \{r,b,g\}$ to denote an arbitrary color. 
On each color disk $\C_v$ we add a boundary component for each edge $e=(u,v)$ incident to $v$. By $\beta_{v, e, k}$ we denote the boundary component corresponding to the vertex $v$, edge $e$, and color $k$.
For each edge $e=(u, v)$ and each pair of boundary components $\beta_{u, e, k_u}$ and $\beta_{v, e, k_v}$ with $k_u \neq k_v$ we construct a tube with boundary components $-\beta_{u, e, k_u}$ and $-\beta_{v, e, k_v}$ denoted $\mathcal{T}_{e, k_u, k_v}$. 
When $k_u = k_v = k$ we construct a tube $\mathcal{T}_{e, k, k}$ and puncture it with a third boundary component $\alpha$ and construct a negatively oriented real projective plane $\mathcal{R}\mathcal{P}_{e, k}$ with boundary $\partial \mathcal{R}\mathcal{P}_{e, k} = -2\alpha$.
We call the resulting complex $\K$ and assign a capacity $c(\sigma)=1$ for every triangle $\sigma$ in $\K$. 

We will show that a maximum integral flow $f$ of $\K$ has $\|f\| = 1$ if and only if $G$ is 3-colorable. The following four properties of a maximum integral flow $f$ imply that $G$ is 3-colorable.
\begin{itemize}
\item $f$ must assign exactly one unit of flow to each triangle in $\mathcal{S}$ since the value of $f$ is equal to $f(\mathcal{S})$.
\item For each vertex $v \in V$ exactly one color disk $\C_v$ is assigned one unit of flow while the other two color disks associated with $v$ are assigned zero units of flow. Otherwise, either conservation of flow is violated or some color disk is assigned a fractional flow value.
\item For each edge $e=(u,v) \in E$ exactly one tube $\mathcal{T}_{e,k_u,k_v}$ with $k_u \neq k_v$ must be assigned one unit of flow with all other tubes associated with $e$ assigned zero units of flow.
The tube $\mathcal{T}_{e, k_u, k_v}$ assigned one unit of flow is the tube connecting the color disks $\C_v$ and $\C_u$ that are assigned one unit of flow by the previous property.

\item $f$ assigns zero flow to every $\mathcal{T}_{e,k,k}$ and $\mathcal{R}\mathcal{P}_{e, k}$ since otherwise the triangles in $\mathcal{R}\mathcal{P}_{e,k}$ would need to have $1/2$ units of flow assigned to them to maintain conservation of flow.
\end{itemize}
These four properties imply that the set of color disks $\{\C_v \mid f(\sigma) = 1,\, \forall \sigma \in \C_v\}$ corresponds to a 3-coloring of $G$.
Conversely, given a 3-coloring of $G$ we assign a flow value of one to each color disk corresponding to the 3-coloring. We extend this assignment to a $\gamma$-flow of value one by assigning a flow value of one to $\mathcal{S}$ and the tubes corresponding to the 3-coloring.
\end{proof}

\subsection{Integral cuts}\label{sec:integral_cuts}
The goal of this section is to show that for simplicial complexes relative torsion-free in dimension $d-1$ there exists optimal solutions to \ref{lp:cut} whose support is a minimum combinatorial $\gamma$-cut.
Note that by Theorem~\ref{thm:dey} a simplicial complex that is relative torsion-free in dimension $d-1$ has a totally unimodular $d$-dimensional boundary matrix. The total unimodularity is key to our proof.
However, we first provide an example of a complex (with relative torsion) whose optimal solution's support does not form a minimum combinatorial $\gamma$-cut.
Our construction is a slight modification of $\M\D$ defined in Section~\ref{sec:integral_flow}.

Consider the simplicial complex constructed by taking $\M\D$ and glueing a wedge sum of two disks $\mathcal{W}$ along the figure-eight $\gamma$. That is, $\partial \mathcal{W} = \gamma$. We give every triangle in the resulting complex a capacity equal to one. A maximum $\gamma$-flow has value $3/2$, so the dual program finds a $\gamma$-cut of the same value. 
One potential optimal solution is a $(d-1)$-cochain whose coboundary assigns a value of $-1/2$ to two triangles in $\mathcal{W}$ and a value of $-1/2$ to one triangle in $\D$. 
The support of this coboundary has weight equal to three, however a minimal combinatorial $\gamma$-cut has weight two by taking only one triangle from $\mathcal{W}$ and one from $\D$. See Figure~\ref{fig:cochain} for an illustration.

\begin{figure}[!htb]
    \centering
    \includegraphics[scale=.5]{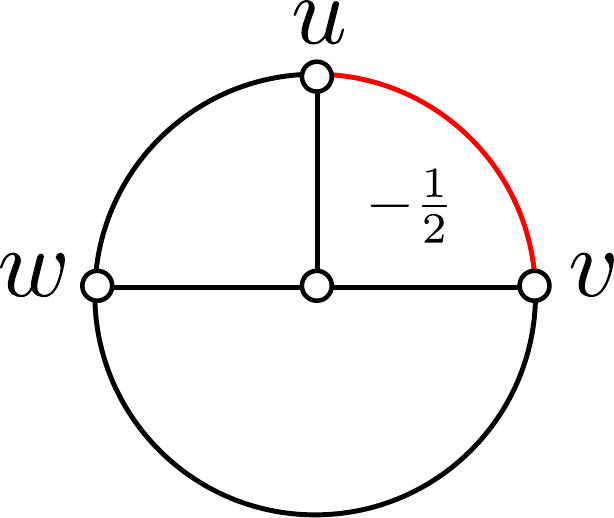}\qquad
    \includegraphics[scale=.5]{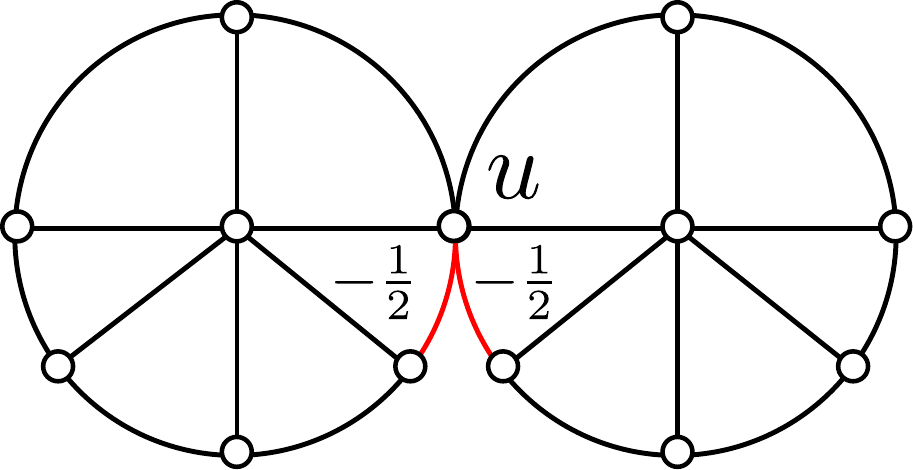}
    \includegraphics[scale=.5]{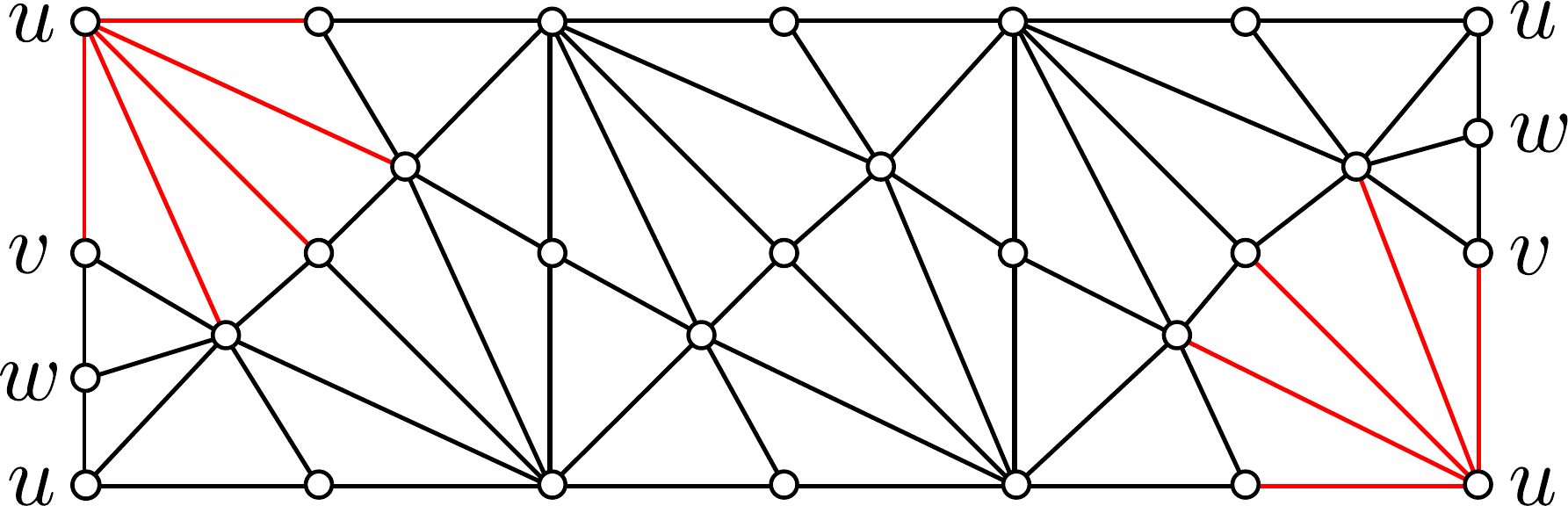}
    \caption{The simplicial complex $\M\D$ with a wedge sum of two disks $\mathcal{W}$ identified to the figure-eight $\gamma$. In red we have a 1-cochain which assigns a value of $-1/2$ to each red edge. The coboundary of the red cochain assigns a value of $-1/2$ to one triangle in $\D$ and a value of $-1/2$ to two triangles in $\mathcal{W}$. The value of the red cochain coincides with the value of the maximum $\gamma$-flow. However, its support is not a minimum combinatorial $\gamma$-cut. A minimum combinatorial $\gamma$-cut picks one triangle from $\D$ and one triangle from $\mathcal{W}$.}
    \label{fig:cochain}
\end{figure}

Now, we show that when $\K$ is relative torsion-free in dimension $d-1$ \ref{lp:cut} has an optimal solution whose support is a minimum directed combinatorial $\gamma$-cut. Specifically, we show that a solution existing on a vertex of the polytope defined by the constraints of \ref{lp:cut} is a cochain $y_{d-1}$ with negative coboundary $y_d$ such that $y_d(\sigma) \in \{0,1\}$ for all $\sigma \in \K^d$ hence $\sum y_d(\sigma) c(\sigma) = \|\supp(y_d)\|$. That is, the value of a vertex solution to \ref{lp:cut} is equal to the cost of $\supp(y_d)$ as a directed combinatorial $\gamma$-cut.

\begin{theorem}\label{thm:totuni}
Let $\K$ be $d$-dimensional simplicial complex that is relative torsion-free in dimension $d-1$ and $y = [y_{d-1}, y_d]^T$ be an optimal vertex solution to the dual program. The set $\supp(y_d)$ is a minimum directed combinatorial $\gamma$-cut.
\end{theorem}
\begin{proof}
We can write the constraint matrix of \ref{lp:cut} as the $2n_d \times (n_d + n_{d-1})$ block matrix
\[A = \begin{bmatrix} \delta & I_{n_d} \\  0_{n_d} & I_{n_d} \end{bmatrix}.\]
Since $\K$ is relative torsion-free in dimension $d-1$ Theorem~\ref{thm:dey} tells us that $\partial_d$ is totally unimodular; further, we have that $\partial^T = \delta$ is also totally unimodular.
Total unimodularity is preserved under the operation of adding a row or column consisting of exactly one component equal to $1$ and the remaining components equal to $0$, so $A$ is totally unimodular~\cite[Section 19.4]{LinearIntegerProgramming}. 
We write \ref{lp:cut} as the linear system $Ax \geq b$ where $b$ is a $n_d + n_{d-1}$ dimensional vector with exactly one component equal to $1$ and the remaining components equal to $0$.
Let $y = [y_{d-1}, y_d]^T$ be an optimal vertex solution to \ref{lp:cut}.
For every $(d-1)$-simplex $\tau \in \K^{d-1}$ we either have $y_{d-1}(\tau) \geq 0$ or $y_{d-1}(\tau) \leq 0$. Let $I'_{n_{d-1}}$ be the matrix whose rows correspond to these inequalities. Note that $I'_{n_{d-1}}$ is a diagonal matrix with entries in $\{-1, 1\}$. Now we consider the $(2n_d + n_{d-1}) \times (n_d + n_{d-1})$ dimensional linear system $A' x \geq b'$ where \[A' = \begin{bmatrix} \delta & I_{n_d} \\ 0 & I_{n_d} \\ I'_{n_{d-1}} & 0 \end{bmatrix}\] 
and $b'$ is constructed by appending extra zeros to $b$. We construct $y'$ from $y$ similarly. 
Note that $A'$ is totally unimodular and $y'$ is a vertex solution of the system. There exists a vertex $v$ of the polyhedron $P \subseteq \R^{n_d + n_{d-1}}$ corresponding to the linear system such that $A'y' = v \geq b'$ such that $n_{d-1} + n_d$ constraints are linearly independent and tight.
Hence, there is an $(n_{d-1} + n_d) \times (n_{d-1} + n_d)$ square submatrix $A''$ with $A''y' = b''$ where $b''$ is $b'$ restricted to the tight constraints.
We will use Cramer's rule to show that the vertex solution $y$ has components coming from the set $\{-1, 0, 1\}$.
Let $A''_{i, b''}$ be the matrix obtained by replacing the $i^{\text{th}}$ column of $A''$ with $b''$. By Cramer's rule we compute the $i^{\text{th}}$ component of $y$ as $y_i = \frac{\det(A''_{i,b''})}{\det(A'')}$. Since both $A''_{i,b''}$ and $A''$ are totally unimodular we have $v_i \in \{-1, 0, 1\}$. Further, we know that $A''$ is non-singular because it corresponds to linearly independent constraints.

By the above argument we know that an optimal solution $y$ to \ref{lp:cut} has all of its components contained in the set $\{-1, 0, 1\}$. The constraint $y_d \geq 0$ means that for all $d$-simplices $\sigma$ we have $y_d(\sigma) \in \{0, 1\}$ and $\sum_{\sigma \in \K^d} y_d(\sigma) c(\sigma) = \|\supp(y_d)\|$. Hence, $\supp(y_d)$ is a minimum directed combinatorial $\gamma$-cut.
\end{proof}

\section{Embedded simplicial complexes}\label{sec:maxflow-embed}
In this section we consider a simplicial flow network $(\K, c, \gamma)$ where $\K$ is a $d$-dimensional simplicial complex with an embedding into $\R^{d+1}$.
Alexander duality implies that $\R^{d+1} \setminus \K$ consists of $\beta_d + 1$ connected components. 
We call these connected components \textit{voids}; exactly one void is unbounded and we denote the voids by $V_i$ for $1 \leq i \leq \beta_{d+1}$.
Given an embedding into $\R^{d+1}$, computing the voids of $\K$ can be done in polynomial time~\cite{Dey2020}.
Further, we assume that the $d$-simplices are consistently oriented with respect to the voids. 
The embedding guarantees that every $d$-simplex $\sigma$ appears on the boundary of at most two voids; by our assumption if $\sigma$ appears on the boundary of two voids then it most be oriented positively on one and negatively on the other. We denote the boundary of the void $V_i$ by $\Bd(V_i)$.
Every $d$-simplex contained in the support of some $d$-cycle is on the boundaries of exactly two voids; it follows that the boundaries of any set of $\beta_d$ voids is a basis of $H_d(\K)$.

In order to state our theorems we need one additional assumption on $\K$. 
We assume there exists some void $V_i$ containing two unit $\gamma$-flows $\Gamma_1,\Gamma_2$ whose supports partition $\Bd(V_i)$: $\supp(\Gamma_1) \cap \supp(\Gamma_2) = \emptyset$ and $\supp(\Gamma_1) \cup \supp(\Gamma_2) = \Bd(V_i)$.
This assumption makes our problem analogous to an $st$-flow network in a planar graph such that $s$ and $t$ appear on the same face. The existence of two unit $\gamma$-flows partitioning the boundary is analogous to the two disjoint $st$-paths on the boundary of the face.
It will be convenient to take the negation of $\Gamma_1$ and treat it as a unit $(-\gamma)$-flow; otherwise the assumption conflicts with the assumed consistent orientation. This is equivalent as it does not change the support of the flow, so for the rest of the section we will take $\Gamma_1$ to be a unit $(-\gamma)$-flow.

From $\K$ we construct its directed dual graph $\K^*$ as follows. Each void becomes a vertex of $\K^*$. 
Each $d$-simplex on the boundary of two voids becomes an edge; since we assumed the $d$-simplices are consistently oriented we direct the dual edge from the negatively oriented void to the positively oriented void. The remaining $d$-simplices only appear on one void and become loops in $\K^*$. For a $d$-simplex $\sigma$ on the boundary of voids $u$ and $v$ we denote its corresponding dual edge $\sigma^* = (u^*, v^*)$ and we weight the edges by the capacity function: $c^*(\sigma^*) = c(\sigma)$.
Let $v_i^*$ be the vertex dual to the void whose boundary is partitioned by $\supp(\Gamma_1)$ and $\supp(\Gamma_2)$. We split $v_i^*$ into two new vertices denoted $s^*$ and $t^*$. The edges incident to $v_i^*$ whose dual $d$-simplices were contained in $\supp(\Gamma_1)$ become incident to $s^*$, and the edges whose dual $d$-simplices were contained in $\supp(\Gamma_2)$ become incident to $t^*$. We add the directed edge $(t^*, s^*)$ and set its capacity to infinity; $c^*((t^*, s^*)) = \infty$.
Returning to the analogy of a planar graph with $s$ and $t$ on the same face, splitting $v_i^*$ is analogous to adding an additional edge from $t$ to $s$ which splits their common face into two. However, for our purposes we are only concerned with the algebraic properties of the construction and do not actually need to modify the simplicial complex.

We need to update the chain complex associated with $\K$ to account for the voids and the splitting of $v_i^*$.
We add an additional basis element $\Sigma$ to $C_d(\K)$ such that $\partial \Sigma = \gamma$ and give it infinite capacity; $c(\Sigma) = \infty$. In our construction $\Sigma$ is dual to the edge $(t^*, s^*)$. In our planar graph analogy $\Sigma$ plays the role of an edge from $t$ to $s$ drawn entirely in the outer face; to make this precise we will need to add an additional chain group $C_{d+1}(\K)$.
We add each void $V_j$ with $j \neq i$ as a basis element of $C_{d+1}(\K)$ and define the boundary operator as $\partial_{d+1}V_j =\sum_{\sigma \in \bd(v)} (-1)^{k_\sigma} \sigma$ where $k_\sigma = 0$ if $\sigma$ is oriented positively on $V_j$ and $k_\sigma = 1$ if $\sigma$ is oriented negatively on $V_j$. 
Next we add additional basis elements $S$ and $T$ whose boundaries are defined by $\partial_{d+1}S = \Gamma_1 + \Sigma$ and $\partial_{d+1}T = \Gamma_2 - \Sigma$.
The inclusion of $C_{d+1}(\K)$ results in a valid chain complex since by definition the image of $\partial_{d+1}$ under each basis element is a $d$-cycle. Moreover, in the new complex we have $H_d(\K) \cong 0$ since the boundaries of the voids generate $H_d(\K)$.

Given our new chain complex we can extend the dual graph $\K^*$ to a dual complex; this construction is reminiscent of the dual of a polyhedron.
We define the dual complex by the isomorphism of chain groups $C_k(\K^*) \cong C_{d - k +1}(\K)$.
The dual boundary operator $\partial^*_k \colon C_k(\K^*) \rightarrow C_{k-1}(\K^*)$ is the coboundary operator $\delta_{d -k +2}$, and the dual coboundary operator $\delta^*_k \colon C_{k - 1}(\K^*) \rightarrow C_k(\K^*)$ is the boundary operator $\partial_{d - k + 2}$. The primal boundary operator commutes with the dual coboundary operator, and the primal coboundary operator commutes with the dual boundary operator. We summarize the duality in the following diagram.

\begin{equation*}
    \begin{tikzcd}
    C_{d+1}(\K) 
        \arrow[d, leftrightarrow, "\cong"] 
        \arrow[r, "\partial_{d+1}"] 
  & C_{d}(\K) 
        \arrow[d, leftrightarrow, "\cong"] 
        \arrow[r, "\partial_{d}"] 
        \arrow[l, "\delta_{d+1}"]
  & \dots 
        \arrow[r, "\partial_1"] 
        \arrow[l, "\delta_d"]
  & C_0(\K) 
        \arrow[l, "\delta_1"]
        \arrow[d, leftrightarrow, "\cong"] \\
    C_0(\K^*) 
        \arrow[r, "\delta^*_1"] 
  & C_1(\K^*) 
        \arrow[l, "\partial^*_1"]
        \arrow[r, "\delta^*_2"] 
  & \dots 
        \arrow[l, "\partial^*_2"]
        \arrow[r, "\delta^*_{d+1}"]
  & C_{d+1}(\K^*)
        \arrow[l, "\partial_{d+1}^*"]
    \end{tikzcd}
\end{equation*}

We now have enough structure to state our duality theorems.
In Section~\ref{sub:maxflow} we show that computing a max-flow for $(\K, c, \gamma)$ is equivalent to computing a shortest path from $s^*$ to $t^*$ in $\K^*$.
In Section~\ref{sub:mincut} we show that computing a minimum cost $\gamma$-cut $p$ is equivalent to computing a minimum cost unit $s^*t^*$-flow in $\K^*$.

\subsection{Max-flow / shortest path duality}\label{sub:maxflow}
We compute a shortest path from $s^*$ to $t^*$ in $\K^*$ using a well-known shortest paths linear program. Details on the linear program can be found in~\cite{jeff}.

\begin{equation}\tag{LP3}
\begin{aligned}
\text{maximize}& \quad \dist(t^*) \\
\text{subject to}& \quad \dist(s^*) = 0\\
& \quad \dist(v^*) - \dist(u^*) \leq c^*((u^*, v^*)) && \forall \, (u^*, v^*) \in E
\end{aligned}
\label{lp:shortest}
\end{equation}
The solution to \ref{lp:shortest} is a function $\dist \colon V(\K^*) \rightarrow \R$ which maps a vertex to its distance from $s^*$ under the weight function $c$. By duality, $\dist$ is a $(d+1)$-cochain mapping the voids to $\R$.
In the following theorem we will show that $\dist$ is equivalent to a $\gamma$-flow with value equal to $\dist(t^*)$.

\begin{theorem}\label{thm:dual1}
Let $(\K,c,\gamma)$ be a simplicial flow network where $\K$ is a $d$-dimensional simplicial complex embedded into $\R^{d+1}$ with two unit $\gamma$-flows whose supports partition the boundary of some void $\Bd(V_i)$.
There is a bijection between $\gamma$-flows of $(\K,c,\gamma)$ and $s^*t^*$-paths in $\K^*$ such that the value of a $\gamma$-flow equals the length of its corresponding $s^*t^*$-path.
\end{theorem}
\begin{proof}
Recall, by our discussion in Section~\ref{sec:cutflow} there is a value-preserving bijection between $\gamma$-flows and $\gamma$-circulations in $\K$ and the value of a $\gamma$-circulation $f$ is given by $f(\Sigma)$.
We can write $f$ as a linear combination of $d$-cycles and by our construction there is a basis for the $d$-cycle space given by $\beta_d$ elements of $\im \partial_{d+1}$.
To form our basis we pick every element of $\im \partial_{d+1}$ except $\partial_{d+1} S$.
Hence, $f$ can be written as a linear combination $f = \sum_{i=0}^{\beta_d - 1} \alpha_i B_i + \alpha_T B_T$ where $\partial_{d+1} V_i = B_i$ and $\partial_{d+1} B_T = T$.
We construct a $(d+1)$-cochain $P$ by the mapping $P(V_i) = \alpha_i$, $P(S) = 0$, $P(T) = \alpha_T$.
Dual to $P$ is a $0$-cochain on the vertices of $\K^*$ which we call $\dist$.
By construction we have $\dist(s^*) = 0$. Since $\K$ is consistently oriented on the voids and $f$ obeys the capacity constraints we have $\dist(v^*) - \dist(u^*) \leq c^*((u^*, v^*))$ for every edge $(u^*, v^*)$.
Finally, we have $\dist(t^*) = f(\Sigma)$ since $f(\Sigma) = P(T) - P(S) = \alpha_T$.

Conversely, let $\dist$ be some vector satisfying the constraints of \ref{lp:shortest}.
As $\dist$ is a cochain on the vertices of $\K^*$ by duality we may view $\dist$ as an element of $C_{d+1}(\K)$ hence $\partial(\dist)$ is a circulation in $\K$ obeying the capacity constraints.
Further, we have that the component of $\partial(\dist)$ indexed by $\Sigma$ is equal to $\dist(t^*) - \dist(s^*) = \dist(t^*)$. Hence, $\|\partial(\dist)\| = \dist(t^*)$ which completes the proof.
\end{proof}

\subsection{Min-cut / min cost flow duality}\label{sub:mincut}

We begin this section by stating the minimum cost flow problem in graphs.
The minimum cost flow problem asks to find the cheapest way to send $k$ units of flow from $s$ to $t$.
An instance of the minimum cost flow problem is a tuple $(G, w, c, k)$ where $G=(V,E)$ is a directed graph, $w, c \in C_1(G)$, and $k \in \R$. The 1-chains represent the weight and capacity of each edge, and $k$ is the demand of the network. The goal of the minimum cost flow problem is to find an $st$-flow obeying the following constraints.
\begin{align*}
\text{minimize}& \quad \sum_{e \in E} w(e)f(e)\\
\text{subject to}& \quad \delta(v) \cdot f = 0 & \forall v \in V \setminus \{s, t\}\\
& \quad 0 \leq f(e) \leq c(e) \\
& \quad \delta(s) \cdot f = -k \\
& \quad \delta(t) \cdot f = k 
\end{align*}
The first two constraints are conservation of flow and capacity constraints. The third and fourth are the demand constraints which say $f$ must send exactly $k$ units of flow from $s$ to $t$; note that only one of these constraints is necessary. 
We will compute a minimum directed $\gamma$-cut in $\K$ by solving the minimum cost flow problem with $k=1$ in $\K^*$.
We assume there is a weight function $w \colon \K^d \rightarrow \R^+$ on the $d$-skeleton of $\K$, which after dualizing becomes a weight function $w^*$ on the edges of $\K^*$.
In the following theorem the capacity function is not needed, so we will assume each edge in $\K^*$ has infinite capacity.

\begin{theorem}\label{thm:dual2}
Let $\K$ be a $d$-dimensional simplicial complex embedded into $\R^{d+1}$ with two unit $\gamma$-flows whose supports partitions the boundary of some void $\Bd(V_i)$.
There is a bijection between $\gamma$-cuts $p$ in $\K$ and unit $s^*t^*$-flows $f$ in $\K^*$ such that $\|p\| = \sum w^*(e)f(e)$.
\end{theorem}
\begin{proof}
Let $p$ be a $\gamma$-cut in $\K$ and let $\bar{p}$ be its negation; that is, $\bar{p}$ is a $(d-1)$-cochain with $\bar{p}(\gamma) = 1$. By construction we have $\delta(\bar{p})(\Sigma) = \bar{p}((\partial \Sigma) = 1$.
We define $f$ to be the image of $\delta(\bar{p})$ under the duality isomorphism. Since $\bar{p}^*$ is a 2-chain in $\K^*$ and $f = \partial^*(\bar{p}^*)$ we see that $f$ is a 1-circulation in $\K^*$.
By removing the edge $\Sigma^* = (t^*, s^*)$ from $f$ we see that $f$ is an $s^*t^*$-flow with $\|f\| = 1$ since $f(\Sigma^*) = \delta(\bar{p})(\Sigma) = 1$.
Finally, we have $\|p\| = \|p'\| = \sum w(\sigma)\delta(p)(\sigma)  = \sum w(\sigma^*)f(\sigma^*)$.

Conversely, let $f^*$ be a 1-circulation in $\K^*$ with $f(\Sigma^*) = 1$. 
By assumption we have $H_d(\K) \cong 0$ and by duality $H_1(\K^*) \cong 0$.
It follows that $f^*$ can be written as a linear combination of boundaries $f^* = \sum \alpha_i B^*_i$ where $B_i^* \in \im \partial^*_2$.
Let $p^*$ be a 2-chain with $\partial_2^* p^* = f^*$. 
Dual to $f^*$ is a $d$-coboundary $f = \sum \alpha_i B_i$ where $B_i \in \im(\delta_d)$ and dual to $p^*$ is a $(d-1)$-cochain $p$ with $\delta(p) = f$.
We have that $\bar{p}$ is a $\gamma$-cut since $\bar{p}(\gamma) = \bar{p}(\partial \Sigma) = \delta(\bar{p})(\Sigma) = -f(\Sigma) = -f^*(\Sigma^*) = -1$.
Finally, we have $\sum w^*(e)f^*(e) = \sum |w(\sigma) \delta(p)(\sigma)| = \|p\| = \|\bar{p}\|$.
\end{proof}

\begin{corollary}
Let $\K$ be a $d$-dimensional simplicial complex embedded in $\R^{d+1}$ with two unit $\gamma$-flows partitioning some $\Bd(V_i)$.
There is a polynomial time algorithm computing a minimum directed combinatorial $\gamma$-cut.
\end{corollary}
\begin{proof}
We solve the minimum cost circulation problem in $\K^*$ setting the demand and every capacity constraint equal to one. The resulting flow is dual to a $\gamma$-cut $p$ in $\K$. Since the minimum cost circulation is integral we have $\|\supp(\delta(p))\| = \|p\|$. That is, the cost of $p$ as a $\gamma$-cut equals the cost of $\supp(\delta(p))$ as a combinatorial $\gamma$-cut.
\end{proof}

\section{Ford-Fulkerson algorithm}\label{sec:ff}
In this section we show how the Ford-Fulkerson algorithm can be used to compute a maximum flow of simplicial flow network $(\K, c, \gamma)$.
In a simplicial flow network the Ford-Fulkerson algorithm picks out a \textit{augmenting chain} at every iteration which is a high dimensional generalization of an augmenting path.
As shown in Section~\ref{sec:integral_flow} a maximum flow of a simplicial flow network with integral capacities may not be integral, so it is not immediate that Ford-Fulkerson is guaranteed to halt.
To remedy this, our implementation of Ford-Fulkerson contains a heuristic reminiscent of the network simplex algorithm.
Our heuristic guarantees that at every iteration of Ford-Fulkerson the flow is a solution on a vertex of the polytope defined by the linear program. Hence, our heuristic makes our implementation of Ford-Fulkerson into a special case of the simplex algorithm.
It follows that Ford-Fulkerson does halt on a simplicial flow network, but the running time may be exponential.
Our heuristic for picking augmenting chains takes $O(n^{\omega +1})$ time since it requires solving $O(n)$ linear systems, each taking $O(n^\omega)$ time using standard methods~\cite{LUP}. In Section~\ref{sec:FFtime} we show that solving a particular type of linear system reduces to finding an augmenting chain, giving us an imprecise lower bound on the complexity of our heuristic. 

We begin this section by defining the concepts of the \textit{residual complex} and of \textit{augmenting chains} which serve as high dimensional generalizations of the residual graph and augmenting paths used in the Ford-Fulkerson algorithm for graphs.
We show that a flow is maximum if and only if its residual complex contains no augmenting chains, generalizing the well-known graph theoretic result.
This work is an extension of previous work done by Latorre who showed one direction of the theorem and leaving the other open~\cite{latorre2012maxflow}.

\subsection{The residual complex}
We now present our definitions of the residual complex and an augmenting chain.

\begin{definition}[Residual complex]
Let $(\K, c, \gamma)$ be a simplicial flow network and $f$ be a feasible flow on the network.
We define a new simplicial flow network called the \textbf{residual complex} to be the tuple $(\K_f, c_f, \gamma)$ constructed in the following way. The $d$-skeleton of $\K_f$ is the union $\K^d \cup -\K^d$, that is, for each $d$-simplex $\sigma$ in $\K$ we add an additional $d$-simplex $-\sigma$ whose orientation is opposite of $\sigma$. $\K_f^{d'} = \K^{d'}$ for dimensions $d' < d$.
The \textbf{residual capacity function} $c_f \colon \K^d_f \rightarrow \R$ is given by \[c_f(\sigma) = \begin{cases} c(\sigma) - f(\sigma) & \sigma \in \K^d,\\ f(\sigma) & -\sigma \in \K^d. \end{cases} \]
\end{definition}

\begin{definition}[Augmenting chain]

Let $\K_f$ be a residual complex for the simplicial flow network $(\K, c, \gamma)$. An \textbf{augmenting chain} is a $d$-chain $\Gamma \in C_d(\K_f)$ such that $\Gamma = \sum \alpha_i \sigma_i$ and $\partial \Gamma = \gamma$ with $\alpha_i \geq 0$.
\end{definition}

Note that an augmenting chain need not obey the residual capacity constraint $c_f$. This is because after finding an augmenting chain the amount of flow sent through the chain will be normalized by the coefficients $\alpha_i$ producing a new chain respecting the capacity constraints.
The following two lemmas prove the main result of the section. The first of which was observed by Latorre~\cite{latorre2012maxflow}.

\begin{lemma}[Latorre~\cite{latorre2012maxflow}]\label{lem:FF1}
Let $(\K, c, \gamma)$ be a simplicial flow network.
If $f$ is a maximum flow then $\K_f$ contains no augmenting chains.
\end{lemma}
\begin{proof}
Let $\Gamma = \sum \alpha_i \sigma_i$ be an augmenting chain in $\K_f$ and let $\alpha = \min \{\frac{1}{\alpha_i}c_f(\sigma_i)\}$.
Define the new flow as follows \[f'(\sigma) = \begin{cases} f(\sigma) + \alpha \cdot \alpha_i & \sigma = \sigma_i, \\ f(\sigma) - \alpha \cdot \alpha_i & -\sigma = \sigma_i. \end{cases} \]
We have $f(\sigma_i) + \alpha \cdot \alpha_i \leq f(\sigma_i) + c_f(\sigma_i) = c(\sigma_i)$ and $f(\sigma_i) - \alpha \cdot \alpha_i \geq f(\sigma_i) - c_f(\sigma_i) = 0$ so $f'$ obeys the capacity constraints.
To show that $f'$ obeys conservation of flow we compute the following equality \[\sum_{\sigma \in \K^d} f'(\sigma) \partial(\sigma) = \sum_{\sigma \in \K^d \setminus \supp(\Gamma)} f(\sigma)\partial(\sigma) + \sum_{\sigma_i \in \supp(\Gamma)} (f(\sigma_i) \pm \alpha \cdot \alpha_i) \partial(\sigma_i) = \alpha \cdot \partial \left( \sum_{\sigma_i \in \supp(\Gamma)} \alpha_i \sigma_i \right) = 0. \] 
\end{proof}

\begin{lemma}\label{lem:FF2}
Let $(\K, c, \gamma)$ be a simplicial flow network. If $f$ is a flow such that $\K_f$ contains no augmenting chains then $f$ is a maximum flow.
\end{lemma}
\begin{proof}
By way of contradiction assume that $f$ is not a maximum flow and $g$ is some other flow with higher value than $f$. We show that the $d$-chain $g - f$ is an augmenting chain in $\K_f$. First, we note that the boundary of $g - f$ is equal to $(\|g\| - \|f\|)\gamma$ implying $g - f$ obeys conservation of flow, so all that remains is to check that $g - f$ obeys the capacity constraints in $\K_f$. Let $\sigma$ be a $d$-simplex in $\K$, there are two cases to consider. First, if $f(\sigma) \leq g(\sigma)$ we have $g(\sigma) - f(\sigma) \leq c(\sigma) - f(\sigma) = c_f(\sigma)$ and the capacity constraint is obeyed.
Second, if $g(\sigma) < f(\sigma)$ we interpret this as applying $|g(\sigma) - f(\sigma)|$ flow to $-\sigma$. Hence, $|g(\sigma) - f(\sigma)| = f(\sigma) - g(\sigma) < f(\sigma) = c_f(\sigma)$ which concludes the proof.
\end{proof}

Lemmas~\ref{lem:FF1} and \ref{lem:FF2} give us the main theorem of the section.

\begin{theorem}\label{thm:residual}
Let $(\K,c, \gamma)$ be a simplicial flow network. A flow $f$ is a maximum flow if and only if $\K_f$ contains no augmenting chains.
\end{theorem}

\subsection{Augmenting chain heuristic}
In this section we provide a heuristic for the Ford-Fulkerson algorithm that is guaranteed to halt on a simplicial flow network.
Our example in Section~\ref{sec:integral_flow} shows that a maximum flow may have fractional value, so it's not immediately clear that Ford-Fulkerson halts on all simplicial flow networks.
To remedy this our heuristic ensures that at each step the flow corresponds to a vertex of the flow polytope (defined in the next paragraph). 
As there are a finite number of vertices, and the value of the flow increases at every step, it follows that under this heuristic Ford-Fulkerson must halt.
Under our heuristic Ford-Fulkerson becomes a special case of the simplex algorithm.
Our heuristic is reminiscent of the network simplex algorithm which maintains a tree at every iteration. See the book by Ahuja, Magnanti, and Orlin for an overview of the network simplex algorithm~\cite{Ahuja}.

We define the \textit{flow polytope} of $(\K, c, \gamma)$ to be the polytope $P \subset \R^{n_d}$ defined by the constraints of the maximum flow linear program \ref{lp:flow}.
A \textit{vertex} of the polytope $P$ is any feasible solution to \ref{lp:flow} with at least $n_d$ tight linearly independent constraints.
We will ensure that at every step of Ford-Fulkerson our flow $f$ is a vertex of $P$. 
To do this we will make sure that the $d$-simplices corresponding to non-tight constraints of \ref{lp:flow} form an acyclic complex. 
Some straightforward algebra implies that this condition is enough to make at least $n_d$ constraints tight.
Let $\H_f$ be the subcomplex of $d$-simplices \say{half-saturated} by $f$; that is, $\sigma \in \H_f$ if and only if its capacity constraint is a strict inequality: $0 < f(\sigma) < c(\sigma)$. 
The half-saturated simplices do not make either of their two corresponding constraints tight, while $d$-simplices not in $\H_f$ make exactly one of their corresponding constraints tight.
We require that $\H_f$ be an acyclic complex at each step of Ford-Fulkerson. 
In the case of graphs, this just means that $\H_f$ is a forest. 
For a $d$-dimensional complex it means that $H_d(\H_f) = 0$. 
Acyclic complexes have been studied by Duval, Klivans, and Martin who show that they share many properties with forests and trees in graphs~\cite{Duval2016}.
The following lemma shows that if $\H_f$ is acyclic then $f$ is a vertex of the flow polytope.

\begin{lemma}\label{lem:vertex}
Let $f$ be a feasible flow for the $d$-dimensional simplicial flow network $(\K, c, \gamma)$.
If the subcomplex of half-saturated $d$-simplices $\H_f$ is acyclic then $f$ is a vertex of the flow polytope $P$.
\end{lemma}
\begin{proof}
In order for $f$ to be a vertex of $P$ we need to show that at least $n_d$ of the constraints of \ref{lp:flow} are tight, and that these $n_d$ tight constraints are linearly independent.
As $f$ is a flow the $2n_{d-1}$ constraints ensuring conservation of flow are always tight.
It follows that we have $n_d - \beta_d$ tight linearly independent constraints corresponding to a basis for $\im \delta_d$.
As $\H_f$ is acyclic we have that $|\supp(\H_f)| \leq n_d - \beta_d$ since at least one $d$-simplex from each basis element of $H_d(\K)$ must be missing from $\H_f$. 
This implies that at most $2(n_d - \beta_d)$ of the $2n_d$ capacity constraints are not tight; equivalently, at least $2n_d - 2(n_d - \beta_d) = 2\beta_d$ of the capacity constraints are tight.
Since $\H_f$ is acyclic we can pick a set of $\beta_d$ $d$-simplices $\Sigma$ such that $\dim H_d(\K \setminus \Sigma) = 0$ and $\dim H_{d-1}(\K \setminus \Sigma) = \dim H_{d-1}(\K)$. Each $d$-simplex $\sigma \in \Sigma$ corresponds to some tight capacity constraint, and since removing $\sigma$ does not change the dimension of $H_{d-1}$ it is not contained in $\im \delta_d$.
It follows that the tight capacity constraint corresponding to $\sigma$ is linearly independent from the $n_d - \beta_d$ conservation of flow constraints.
Finally, since each tight capacity constraint corresponds to a unique $\sigma$ they are all linearly independent from each other.
\end{proof}

At each iteration of Ford-Fulkerson we want to pick an augmenting chain such that the resulting flow leaves $\H_f$ acyclic.
It's not clear how to pick such an augmenting chain. However, no matter what augmenting chain we pick we can always repair the flow in a way that the resulting flow leaves $\H_f$ acyclic.
We describe our method for repairing the flow in the following lemma.

\begin{lemma}\label{lem:repair}
Let $f$ be a feasible flow for the $d$-dimensional simplicial flow network $(\K, c, \gamma)$. If the subcomplex of half-saturated $d$-simplices $\H_f$ is not acyclic then in $O(n^{\omega+1})$ time we can construct a new flow $f'$ such that $\H_{f'}$ is acyclic and $\|f\| = \|f'\|$.
\end{lemma}
\begin{proof}
Let $f$ be a feasible flow such that $\H_f$ is not acyclic.
In polynomial time we compute a basis for $H_d(\H_f)$.
Let $\Sigma = \sum \alpha_i \sigma_i$ be a basis element of $H_d(\H_f)$ and let $\alpha = \min \{ \frac{1}{\alpha_i} c_f(\sigma_i)\}$. As in Lemma~\ref{lem:FF1} we construct the new flow $f'$ by \[f' = \begin{cases} f(\sigma) + \alpha \cdot \alpha_i & \sigma = \sigma_i \\ f(\sigma) - \alpha \cdot \alpha_i & -\sigma = \sigma_i \end{cases}. \]
The new flow $f'$ saturates some $d$-simplex in $\Sigma$ and does not introduce any new $d$-cycles to $\H_f$ as it only affects the half-saturated edges. We call the new subcomplex of half-saturated simplices $\H_{f'}$ and observe that $\dim H_d(\H_{f'}) < \dim H_d(\H_f)$.
Since $f'$ is constructed by adding a $d$-cycle to $f$ we have that $\|f\| = \|f'\|$.
We repeat the process of computing a homology basis for $\H_{f'}$ and saturating some basis element until $\H_{f'}$ is acyclic.

It remains to compute the running time of the above procedure.
Computing a homology basis takes $O(n^{\omega})$ time . To repair the flow we need to make at most $O(n)$ homology basis computations, hence the total running time is $O(n^{\omega + 1})$.
\end{proof}

To wrap up the section, we state our main theorem whose proof is immediate from Lemmas~\ref{lem:vertex} and \ref{lem:repair}.

\begin{theorem}\label{thm:heuristic}
Given a simplicial flow network $(\K, c, \gamma)$ we can compute a maximum flow $f$ by using the Ford-Fulkerson algorithm with the following heuristic: at every iteration pick an augmenting chain such that the subcomplex of half-saturated $d$-simplices $\H_f$ is acyclic.
\end{theorem}

\subsection{Lower bounds}\label{sec:FFtime}
We have shown a heuristic for which given a simplicial flow network $(\K,c,\gamma)$ Ford-Fulkerson is guaranteed to halt with a maximum flow.
In the worst case our algorithm runs in exponential time. 
In this section we focus our attention on the time required to find an augmenting chain.
The running time of our heuristic is determined by the time it takes to compute the augmenting chain and repair the flow at each iteration of the algorithm.
Computing the augmenting chain takes $O(n^{\omega})$ time by solving the linear system. 
Repairing the flow takes $O(n^{\omega +1})$ time since it requires $O(n)$ homology basis computations which each take time $O(n^{\omega})$.
We show that it is unlikely that this running time can be substantially improved. 
More specifically, we show that finding a non-negative solution to a linear system $Ax = b$ when $A$ has entries in $\{-1,0,1\}$ reduces to computing an augmenting chain for $(\K, c, \gamma)$, hence the complexity of solving a linear system in this form serves as a lower bound on computing an augmenting chain.
Given a linear system we construct a 2-complex with a 1-cycle $\gamma$ such that finding a 2-chain $\Gamma$ with non-negative coefficients and $\partial \Gamma = \gamma$ is equivalent to finding a non-negative feasible solution to the linear system. 
Further, the complex $\K$ used in the reduction is relatively torsion-free, so the total unimodularity of its boundary matrix cannot be used to speed up the computation.
Our reduction is essentially the same as one used by Chen and Freedman to show homology localization over $\Z_2$ is \NP-hard \cite{Chen2010}. However, we modify it slightly since we consider coefficients over $\R$.
We only give a proof sketch as our reduction is almost identical to that of Chen and Freedman's.

\begin{theorem}\label{thm:lowerbound}
Let $Ax = b$ be a linear system where $A$ has entries in $\{-1,0,1\}$. In polynomial time we can construct a 2-dimensional, relatively torsion-free, simplicial complex $\K$ and a 1-cycle $\gamma$ such that if a 2-chain $\Gamma$ is an augmenting chain for $\gamma$ then it is a non-negative solution to $Ax = b$.
\end{theorem}
\begin{proof}
We construct a cell complex $\K$ from $A$ as follows. For each of the $m$ rows we construct a 1-cycle $C_i$.
For each column vector $v_j$ we construct a punctured sphere $T_j$ with boundary components $C_{i,j},$ for each $v_{i,j} = 1$ and $-C_{i,j}$ for each $v_{i,j} = -1$.
Define the 1-cycle $\gamma$ to be $b$ with respect to the basis given by the boundary components of $\K$.
By our construction a vector $x = (x_1,\dots,x_n)$ is a feasible solution to $Ax = b$ if and only if the 2-chain $\sum x_i T_i$ has boundary $\gamma$. Hence, computing an augmenting chain for $\gamma$ is equivalent to computing a non-negative solution to the linear system.
It remains to show that $\K$ can be triangulated into a simplicial complex. We refer the reader to \cite{Chen2010} for a triangulation. To see that $\K$ is relatively torsion-free we refer to reader to \cite{Dey2011} which characterizes relative torsion in 2-complexes by forbidding certain \mobius{} subcomplexes.
\end{proof}

\bibliographystyle{plain}
\bibliography{flows}

\begin{thebibliography}{10}

\bibitem{Ahuja}
Ravindra~K. Ahuja, Thomas~L. Magnanti, and James~B. Orlin.
\newblock {\em Network Flows: Theory, Algorithms, and Applications}.
\newblock Prentice-Hall, Inc., USA, 1993.

\bibitem{AlexanderDuality}
J.~W. Alexander.
\newblock A proof and extension of the {Jordan-Brouwer} separation theorem.
\newblock {\em Transactions of the American Mathematical Society},
  23(4):333--349, 1922.

\bibitem{Borradaile2009}
Glencora Borradaile and Philip Klein.
\newblock An {$O(n \log n)$} algorithm for maximum {$st$}-flow in a directed
  planar graph.
\newblock {\em J. ACM}, 56(2), April 2009.

\bibitem{cen-mcshc-09}
Erin~W. Chambers, Jeff Erickson, and Amir Nayyeri.
\newblock Minimum cuts and shortest homologous cycles.
\newblock In {\em Proc. 25th Ann. Symp. Comput. Geom.}, pages 377--385, 2009.

\bibitem{HomFlow}
Erin~W. Chambers, Jeff Erickson, and Amir Nayyeri.
\newblock Homology flows, cohomology cuts.
\newblock {\em SIAM Journal on Computing}, 41(6):1605--1634, 2012.

\bibitem{Chen2010}
Chao Chen and Daniel Freedman.
\newblock Hardness results for homology localization.
\newblock In {\em Proceedings of the Twenty-First Annual ACM-SIAM Symposium on
  Discrete Algorithms}, SODA '10, page 1594–1604, USA, 2010. Society for
  Industrial and Applied Mathematics.

\bibitem{Dey2011}
Tamal~K. Dey, Anil~N. Hirani, and Bala Krishnamoorthy.
\newblock Optimal homologous cycles, total unimodularity, and linear
  programming.
\newblock {\em SIAM J. Comput.}, 40(4):1026--1044, July 2011.

\bibitem{Dey2020}
Tamal~K. Dey, Tao Hou, and Sayan Mandal.
\newblock Computing minimal persistent cycles: Polynomial and hard cases.
\newblock In {\em Proceedings of the Thirty-First Annual ACM-SIAM Symposium on
  Discrete Algorithms}, SODA '20, page 2587–2606, USA, 2020. Society for
  Industrial and Applied Mathematics.

\bibitem{Duval2015}
Art~M. Duval, Caroline~J. Klivans, and Jeremy~L. Martin.
\newblock Cuts and flows of cell complexes.
\newblock {\em Journal of Algebraic Combinatorics}, 41:969--999.

\bibitem{Duval2016}
Art~M. Duval, Caroline~J. Klivans, and Jeremy~L. Martin.
\newblock {\em Simplicial and cellular trees}, pages 713--752.
\newblock Springer International Publishing, Cham, 2016.

\bibitem{jeff}
Jeff Erickson.
\newblock {\em Algorithms}.
\newblock \url{http://algorithms.wtf}, 2019.

\bibitem{en-mcsnc-11}
Jeff Erickson and Amir Nayyeri.
\newblock Minimum cuts and shortest non-separating cycles via homology covers.
\newblock In {\em Proc. 22nd Ann. ACM-SIAM Symp. Discrete Algorithms}, pages
  1166--1176, 2011.

\bibitem{VeinottJr1968}
Arthur {F. Veinott, Jr}. and George~B. Dantzig.
\newblock Integral extreme points.
\newblock {\em {SIAM} Review}, 10(3):371--372, July 1968.

\bibitem{f-faspp-87}
Greg~N. Frederickson.
\newblock Fast algorithms for shortest paths in planar graphs with
  applications.
\newblock {\em SIAM J. Comput.}, 16(6):1004--1004, 1987.

\bibitem{Hassin1981MaximumFI}
Refael Hassin.
\newblock Maximum flow in {$(s, t)$} planar networks.
\newblock {\em Inf. Process. Lett.}, 13:107, 1981.

\bibitem{hj-oamfu-85}
Refael Hassin and Donald~B. Johnson.
\newblock An {$O(n\log^2 n)$} algorithm for maximum flow in undirected planar
  networks.
\newblock {\em SIAM J. Comput.}, 14(3):612--624, 1985.

\bibitem{Hatcher}
Allen Hatcher.
\newblock {\em {Algebraic topology}}.
\newblock Cambridge Univ. Press, Cambridge, 2000.

\bibitem{MOhardness}
Sergei~Ivanov (https://mathoverflow.net/users/4354/sergei ivanov).
\newblock computational complexity.
\newblock MathOverflow.
\newblock URL:https://mathoverflow.net/q/118428 (version: 2013-01-09).

\bibitem{LUP}
Oscar~H Ibarra, Shlomo Moran, and Roger Hui.
\newblock A generalization of the fast {LUP} matrix decomposition algorithm and
  applications.
\newblock {\em Journal of Algorithms}, 3(1):45--56, 1982.

\bibitem{is-mfpn-79}
Alon Itai and Yossi Shiloach.
\newblock Maximum flow in planar networks.
\newblock {\em SIAM J. Comput.}, 8:135--150, 1979.

\bibitem{insw-iamcmf-11}
Giuseppe~F. Italiano, Yahav Nussbaum, Piotr Sankowski, and Christian
  Wulff-Nilsen.
\newblock Improved algorithms for min cut and max flow in undirected planar
  graphs.
\newblock In {\em Proc. 43rd Ann. ACM Symp. Theory Comput.}, pages 313--322,
  2011.

\bibitem{latorre2012maxflow}
Fabian Latorre.
\newblock The maxflow problem and a generalization to simplicial complexes,
  2012.

\bibitem{hodge_survey}
Lek-Heng Lim.
\newblock Hodge {Laplacians} on graphs.
\newblock {\em SIAM Review}, 63(3):685--715, 2020.

\bibitem{morell2020minimumcost}
Sarah Morell, Ina Seidel, and Stefan Weltge.
\newblock Minimum-cost integer circulations in given homology classes, 2020.

\bibitem{r-mstcp-83}
John Reif.
\newblock Minimum $s$-$t$ cut of a planar undirected network in {$O(n\log^2
  n)$} time.
\newblock {\em SIAM J. Comput.}, 12:71--81, 1983.

\bibitem{LinearIntegerProgramming}
Alexander Schrijver.
\newblock {\em Theory of Linear and Integer Programming}.
\newblock {John Wiley \& Sons, Inc.}, USA, 1986.

\end{thebibliography}

\appendix
\section{Directed paths in simplicial complexes}
In a directed graph an $st$-path is equivalent to an 1-chain $P = \sum \alpha_i e_i$ with $\partial P = t - s$ such that $\alpha_i \in \{0, 1\}$.
The generalization to simplicial complexes is straightforward: we define a $d$-dimensional directed $\gamma$-path in a $d$-complex $\K$ to be a $d$-chain $P = \sum \alpha_i \sigma_i$ such that $\partial P = \gamma$ and $\alpha_i \in \{0, 1\}$ for all $i$.
We will now show that computing a directed $\gamma$-path is \NP-complete for $d \geq 2$. The reduction from graph 3-coloring is a slight adaptation of the proof of \autoref{thm:inthard}.

\begin{theorem}\label{thm:aughard}
Computing a directed $\gamma$-path in a $d$-dimensional simplicial complex is \NP-complete for $d \geq 2$.
\end{theorem}
\begin{proof}
First, to show that computing a directed $\gamma$-path is in \NP~we note that we can check the coefficients and the boundary of a 2-chain in polynomial time. To prove \NP-hardness we give a reduction from graph 3-coloring.

Given a graph $G = (V, E)$ we construct a 2-complex $\K$ with a boundary component $\gamma$ such that there exists a directed $\gamma$-path if and only if $G$ is 3-colorable.
First, construct a punctured sphere with $|V| + 1$ boundary components. One of these boundary components is $\gamma$, the remaining $|V|$ are in bijection with the vertices of $G$ and we will denote the component corresponding to the vertex $v$ as $\gamma_v$.
For each vertex $v$ we construct three additional punctured spheres $v_r, v_b, v_g$ each with $\deg(v) + 1$ boundary components. These punctured spheres correspond to the three potential colors of $v$: red, blue, and green and we refer to them as the \emph{color surfaces} of $v$. We glue $v_r, v_b, v_g$ to $\gamma_v$ each along some boundary component.
For each edge $e=(u, v)$ we construct nine tubes with two boundary components. Each tube connects a color surface of $u$ to a color surface of $v$, for example the tube $\mathcal{T}_{r, b}$ connects the red color surface of $u$ with the blue color surface of $v$.
We orient the complex such that each 2-chain with boundary $\gamma$ and coefficients in $\{0, 1\}$ is an oriented manifold with boundary $\gamma$.
For each tube connecting two color surfaces of the same color (for example, $\mathcal{T}_{r, r}$) we invert the orientation of one simplex, such that any bounding chain for $\gamma$ including the inverted simplex will have to assign it a coefficient of -1.

Any 2-chain with boundary $\gamma$ and coefficients in $\{0, 1\}$ is a surface as for each vertex it must contain exactly one color surface and for each edge it must contain exactly one tube. Moreover, if any tube connecting two color surfaces of the same color is contained in the solution it must contain some simplex a negative coefficient and is not an directed $\gamma$-path. It follows that $G$ is 3-colorable if and only if there is some directed $\gamma$-path.
\end{proof}
\end{document}